\documentclass[aos]{imsart}

\RequirePackage{amsthm,amsmath,amsfonts,amssymb}
\RequirePackage[numbers]{natbib}
\RequirePackage[colorlinks,citecolor=blue,urlcolor=blue]{hyperref}
\RequirePackage{graphicx}
\RequirePackage{enumitem}
\usepackage[notextcomp]{kpfonts}

\startlocaldefs
\theoremstyle{plain}
\newtheorem{theorem}{Theorem}
\newtheorem{proposition}{Proposition}

\theoremstyle{remark}
\newtheorem{definition}{Definition}
\newtheorem{assumption}{Assumption}
\DeclareMathOperator*{\argmax}{arg\,max}

\newcommand{\bd}[1]{\boldsymbol{#1}}
\newcommand{\bX}{\bd{X}}

\newcommand{\bbeta}{\bd{\beta}}
\newcommand{\bmu}{\bd{\mu}}
\newcommand{\bSigma}{\bd{\Sigma}}
\newcommand{\Dc}{\mathcal{D}}
\newcommand{\suml}[1]{\sum\limits_{#1}}
\newcommand{\bSigmasi}{\bSigma_{s}}

\newcommand{\sia}{\sigma_{s0}}
\newcommand{\sib}{\sigma_{s1}}

\newcommand{\bgamma}{\bd{\gamma}}
\newcommand{\balpha}{\bd{\alpha}}
\newcommand{\btheta}{\bd{\theta}}
\newcommand{\bx}{\bd{x}}
\newcommand{\bphi}{\bd{\phi}}
\newcommand{\beia}{\bd{e}_{i0}}
\newcommand{\beib}{\bd{e}_{i1}}
\newcommand{\etay}{\eta_{y}}
\newcommand{\etasa}{\eta_{s0}}
\newcommand{\etasb}{\eta_{s1}}
\newcommand{\nuy}{\nu_{y}}
\newcommand{\nusa}{\nu_{s0}}
\newcommand{\nusb}{\nu_{s1}}

\newcommand{\sigmasa}{\sigma_{s0}}
\newcommand{\sigmasb}{\sigma_{s1}}
\newcommand{\sigmay}{\sigma_y}
\newcommand{\byaa}{\beta_{00}}
\newcommand{\byab}{\beta_{01}}
\newcommand{\byba}{\beta_{10}}
\newcommand{\bybb}{\beta_{11}}

\newcommand{\sigmasat}{{\sigma_{s0}^{\ast}}}
\newcommand{\sigmasbt}{{\sigma_{s1}^{\ast}}}

\newcommand{\byabt}{{\beta_{01}^{\ast}}}
\newcommand{\bybat}{{\beta_{10}^{\ast}}}

\newcommand{\rhot}{{{\rho}^{\ast}}}
\newcommand{\sigmayt}{{{\sigma}_y^{\ast}}}
\newcommand{\zetat}{{{\zeta}^{\ast}}}
\newcommand{\psit}{{{\psi}^{\ast}}}

\newcommand{\Rs}{\mathbb{R}}

\newcommand{\ind}{\perp\!\!\!\!\perp} 

\newcommand{\Var}{\text{Var}}
\newcommand{\Cov}{\text{Cov}}
\newcommand{\Cor}{\text{Cor}}

\endlocaldefs

\begin{document}

\begin{frontmatter}
\title{Partial identification of principal causal effects under violations of principal ignorability}

\begin{aug}
\author[A]{\fnms{Minxuan }~\snm{Wu}\ead[label=e1]{wuminxuan@ufl.edu}}
\and\author[A]{\fnms{Joseph}~\snm{Antonelli}\ead[label=e2]{jantonelli@ufl.edu}}
\address[A]{Department of Statistics, University of Florida\printead[presep={,\ }]{e1,e2}}

\end{aug}

\begin{abstract}
Principal stratification is a general framework for studying causal mechanisms involving post-treatment variables. When estimating principal causal effects, the principal ignorability assumption is commonly invoked, which we study in detail in this manuscript. Our first key contribution is studying a commonly used strategy of using parametric models to jointly model the outcome and principal strata without requiring the principal ignorability assumption. We show that even if the joint distribution of principal strata is known, this strategy necessarily leads to only partial identification of causal effects, even under very simple and correctly specified outcome models. While principal ignorability leads to point identification in this setting, we discuss alternative, weaker assumptions and show how they can lead to informative partial identification regions. An additional contribution is that we provide theoretical support to strategies used in the literature for identifying association parameters that govern the joint distribution of principal strata. We prove that this is possible, but only if the principal ignorability assumption is violated. Additionally, due to partial identifiability of causal effects even when these association parameters are known, we show that these association parameters are only identifiable under strong parametric constraints. Lastly, we extend these results to more flexible semiparametric and nonparametric Bayesian models.
\end{abstract}


\end{frontmatter}


\section{\label{sec:introduction} Introduction}

In many causal inference problems, there exist post-treatment (also referred to as intermediate) variables, which can pose challenges for identification and inference. Principal stratification \citep{frangakis_rubin_2002} is a general framework for studying causal mechanisms involving such post-treatment variables. It formally defines causal estimands of interest in terms of principal strata, where principal strata represent the joint potential values of the intermediate variable. The principal strata are not affected by the treatment and can therefore be treated as pre-treatment covariates, which leads to well-defined causal estimands conditional on the principal strata. These are local causal effects in the sense that they are defined for sub-populations of units defined by principal strata. Many scientific questions can be addressed within this framework, such as mediation \citep{vanderweele2008simple, mealli2012refreshing,kim2019bayesian}, noncompliance \citep{imbens_rubin_1997, mealli2013using}, surrogate endpoints \citep{zigler2012bayesian}, or truncation by death \citep{zhang2003estimation, ding2011identifiability, yang2016using}. 

In settings with binary treatments, principal strata are represented by the joint values of potential intermediates under both treatment levels. At most we get to observe only one of these two potential intermediates, which means that we require additional structural assumptions in order to identify principal causal effects (PCEs). Monotonicity and exclusion restriction assumptions, which reduce the number of principal strata and place restrictions on causal effects in certain strata, are commonly used, particularly in settings with noncompliance in randomized trials \citep{angrist_imbens_rubin_1996, hirano_imbens_rubin_zhou_2000, rubin_2006, baccini_mattei_mealli_2017, bia2022assessing}. Principal ignorability is another commonly invoked assumption \citep{jo_stuart_2009, jiang2021identification, mattei2023assessing}, which enforces that the potential outcome under one treatment level only depends on the potential intermediate at that same treatment level. All of these assumptions are fundamentally not testable, and their plausibility should be judged within the context of the problem being studied. 

Given the reliance on key untestable assumptions, many researchers have worked to develop approaches accounting for potential violations of causal assumptions. One approach is to perform sensitivity analysis, which highlights how results would change under specific violations of core assumptions, such as the principal ignorability assumption \citep{ding_lu_2016, wang_zhang_mealli_bornkamp_2022, trang_quynh_nguyen_stuart_scharfstein_ogburn_2024}. There is also an extensive literature on partial identification for principal causal effects \citep{cheng2006bounds, imai2008sharp, lee2009training, mealli2013using, mealli2016identification, yang2016using} where causal effects can be bounded under weaker assumptions than those required for point identification. Most of these previous studies have focused on randomized studies with binary intermediate variables. Others have focused on specific situations such as those with truncation by death, or when a secondary outcome is observed, which can provide more informative bounds. We focus on more general situations throughout this manuscript, or when such additional information is not available. 

Another line of research aims to incorporate parametric assumptions to obtain identification when nonparametric identification is difficult or impossible to obtain. This is particularly the case with continuous intermediate variables, which we address in this manuscript, as standard assumptions such as monotonicity and exclusion restrictions are not sufficient for nonparametric identification. Recent works studying nonparametric identification in this more difficult scenario have made the principal ignorability assumption, along with the strong parametric assumption that the joint distribution of the two potential intermediates is governed by a known copula with a known correlation parameter \citep{jiang2021identification, lu_jiang_ding_2023, zhang_yang_2025}. Results are then estimated across a range of plausible correlation parameters to see if results are robust to this choice. Other works have attempted to estimate this unknown correlation parameter by jointly modeling the potential intermediates and observed outcomes simultaneously \citep{schwartz_li_mealli_2011, bartolucci_grilli_2011}. The idea behind this strategy is that one cannot directly estimate the association between the two potential intermediates because they are never jointly observed, but this information can be obtained from a correctly specified outcome regression model. Related research aims to jointly model the distribution of the potential intermediates and observed outcome using either parametric \citep{jin_rubin_2008} or semiparametric models \citep{kim2019bayesian,kim2024bayesian}, and assume the correlation between the two potential intermediates is known. Many of these model-based approaches do not make the principal ignorability assumption, and they allow the observed outcome to depend on the values of both potential intermediate values. 

In this article, we make a number of contributions to the literature on identification and partial identification of principal causal effects, particularly in relation to the principal ignorability assumption. We focus on the aforementioned model-based approaches that jointly model the potential intermediates and observed outcome as these approaches are commonly used, and as we will show, have inherent difficulties with respect to identification that must be managed carefully. Our first contribution is to discuss the role of principal ignorability in estimating the unknown correlation parameter between the two potential intermediate variables. Recent works have estimated this parameter by jointly modeling the potential intermediates and observed outcome, and we study this process theoretically under a Bayesian model. We show that the posterior distribution of this crucial, and typically unidentified, parameter is indeed consistent for the true correlation parameter, but only if the outcome model is known and principal ignorability is violated. We derive the asymptotic variance of this posterior distribution providing insights on when this parameter can be estimated efficiently. While this is a promising result, extending to the case where the outcome model is estimated is more challenging, as our subsequent results show that even when the correlation parameter is known, an outcome model that allows the outcome to depend on both potential intermediate variables is necessarily only partially identifiable, even under extremely restrictive parametric assumptions. We discuss alternative assumptions that can be made to make these partial identification regions more informative, or even obtain point identification. Lastly, we extend these results to more flexible semiparametric and nonparametric Bayesian models.

\section{Notation and review of existing work}
\label{sec:NotationExisting}
\subsection{Notation and estimands}

Throughout let $\boldsymbol{X}$ denote a $p$-dimensional vector of pre-treatment covariates, $T$ a binary treatment, $S$ a binary or continuous intermediate, and $Y$ an outcome of interest. Our observed data consists of $n$ independent copies of these random variables, and we refer to the observed data by $\mathcal{D}$. We adopt the potential outcome framework \citep{rubin1974estimating} and define $S(t)$ to be the potential intermediate that would be observed had treatment been set to $t$, and similarly let $Y(t)$ be the potential outcome. Throughout, we let $\bd{U} = (S(0), S(1))$ denote the joint values of the two potential intermediates, and let $S$ be the observed value of the intermediate. We assume the Stable Unit Treatment Value Assumption (SUTVA, \cite{rubin_1980}), which implies that 1) there is no 'interference,' meaning that the potential outcomes for unit $i$ do not change with the treatments assigned to other units, and 2) there are not different versions of each treatment level. This assumption ensures that $S = S(T)$ and $Y = Y(T)$. Principal strata are defined by the joint potential values of the intermediate variable $\bd{U}=(S(0),S(1))$. The causal estimands of interest are then defined in terms of average differences in potential outcomes within principal strata, which are given by
\begin{equation*}
    PCE(\bd{u})=E\{Y(1)-Y(0) \mid \bd{U}=\bd{u}\}.
\end{equation*}
One may also be interested in these effects conditional on covariates, which are given by
\begin{equation*}
    PCE(\bd{u},\bx)=E\{Y(1)-Y(0) \mid \bd{U}=\bd{u}, \bd{X} = \bd{x}\}.
\end{equation*}
Now that these effects are defined, we can review standard identification assumptions used to allow these to be identified from the observed data. 

\subsection{Review of identification assumptions in PS analysis}

\label{subsec: review of identification assumptions}
A crucial assumption for identifying causal effects is that the treatment is exchangeable given covariates, or that there are no unmeasured confounders, which is given in the following assumption. 
\begin{assumption}[Treatment Ignorability] \label{assump:Ignorability}
    $T \ind (S(0),S(1),Y(0),Y(1))\mid \bd{X}$.
\end{assumption}
Treatment ignorability is a standard assumption and it holds by design in a randomized experiment. Importantly, it ensures that there are no unmeasured confounders of the treatment-intermediate and treatment-outcome relationship. Additionally, we must make a positivity assumption ensuring that treatment can take either value for any unit in our population.
\begin{assumption}[Positivity] \label{assump:Positivity}
      $0 < P(T = 1 \mid \bd{X} = \bx) < 1$ for all $\bd{x}$.
\end{assumption}
These two assumptions, combined with the SUTVA assumption, allow for identification of average causal effects, but they are generally not sufficient for identification of principal causal effects. In what follows, we detail different identification assumptions for PCEs separated by the binary and continuous intermediate settings, where these assumptions differ. 

\subsection{Binary intermediates}
\label{subsec:BinaryAssumptionReview}
Common assumptions with binary intermediates are monotonicity, exclusion restrictions, and principal ignorability.
\begin{assumption}[Monotonicity] \label{assump:monotonicity}
    $S(1) \geq S(0)$.
\end{assumption}
The monotonicity assumption, sometimes referred to as the no defiers assumption, is plausible in many applications. For instance, in studies of noncompliance, treatment is not available to those who are randomly assigned to the control condition. This assumption rules out the principal strata ($S(0)=1,S(1)=0$), which leads to identification of the distribution of principal strata. This assumption is commonly used along with the following assumption:
\begin{assumption}[Exclusion Restriction] \label{assump:ER} There are two exclusion restrictions:
    \begin{enumerate}
        \item Exclusion restriction for units of type 00: $Y(0)=Y(1)$ if $\bd{U}=(0,0)^T$.
        \item Exclusion restriction for units of type 11: $Y(0)=Y(1)$ if $\bd{U}=(1,1)^T$.
    \end{enumerate}
\end{assumption}
This assumption rules out a direct effect of the treatment on the outcome that does not go through a change in the intermediate. Again, this is reasonable in randomized studies of noncompliance where randomization to a treatment should not affect the outcome, except through its impact on what treatment is actually taken. Under the exclusion restriction and monotonicity assumptions, the average causal effect for compliers is identified as the ratio of the average causal effects on $Y$ and $S$: 
\begin{equation*}
    E\{Y(1)-Y(0) \mid S(0)=0,S(1)=1\}=\frac{E\{Y(1)-Y(0)\}}{E\{S(1)-S(0)\}}.
\end{equation*}
While plausible in certain settings, such as those with randomized treatment assignment and noncompliance, these are frequently not likely to hold, as there are many settings where a direct effect of treatment is plausible. In such settings, one can instead use the following assumption.
\begin{assumption}[Principal Ignorability] \label{assump:PI}
    \begin{equation*}
        E\{Y(t)\mid S(t),S(1 - t),\bX\}=E\{Y(t) \mid S(t),\bX\},\qquad t=0,1.
    \end{equation*}
\end{assumption}
Note that researchers sometimes invoke a slightly stronger assumption of full conditional independence, though conditional mean independence is sufficient for all estimands considered here so we proceed in this manner. One implication of this assumption is that $\bX$ includes all confounders between the intermediate and outcome. Alternatively, the principal ignorability assumption can be viewed as a homogeneity assumption, as it implies that the potential outcome means are the same across certain principal strata. This assumption greatly facilitates identification because under principal ignorability and treatment ignorability, we have that the PCE for principal strata $\bd{u} = (s_0, s_1)$ is identified as
\begin{equation*}
    PCE(\bd{u},\bx)=E(Y\mid S=s_1,\bX=\bx,T=1)-E(Y\mid S=s_0,\bX=\bx,T=0),
\end{equation*}
which is a function of observable variables. The identification of the average PCE requires integration of this quantity over the covariate distribution, which requires knowledge of the conditional probabilities of principal strata, $P(\bd{U}\mid \bX)$, also known as the principal score. The identification of the principal score requires additional assumptions beyond principal ignorability.

\subsection{Continuous intermediates}

For a binary $S$, there are only four principal strata, and the monotonicity assumption eliminates one of these groups and enables the identification of the probability mass of $\bd{U} \mid \bX$ \citep{angrist_imbens_rubin_1996}. For a continuous $S$, there are infinitely many principal strata, and while the monotonicity assumption can rule out some of them, there still remain infinitely many principal strata. Thus, monotonicity alone cannot produce the identification of the distribution of $\bd{U} \mid \bX$, sometimes referred to as the principal density for continuous intermediates. Typically, researchers make parametric assumptions in order to identify this distribution. One approach is to assume the joint distribution of principal strata follows a known copula function $\mathbb{C}_{\rho}(\cdot,\cdot)$ (see \cite{lu_jiang_ding_2023}; \cite{jiang2021identification}; \cite{bartolucci_grilli_2011}):
\begin{equation*}
    pr(S(0) \leq s_0, S(1) \leq s_1\mid\bX)=\mathbb{C}_{\rho}(pr(S(0)\leq s_0\mid\bX),pr(S(1)\leq s_1\mid\bX)),
\end{equation*}
where the parameter $\rho$ indicates the correlation between $S(0)$ and $S(1)$ given $\bd{X}$. 

In general, the explicit identification of the correlation between potential intermediates is infeasible because only one of them is ever observed per observation. This has led to a common practice of $\rho$ being selected based on domain knowledge, and a sensitivity analysis is performed to assess the impact of this parameter. Throughout this article, we refer to parameters like $\rho$ that cannot be identified marginally as the association parameters of the principal strata. Under this distributional assumption, combined with Assumptions (\ref{assump:Ignorability}) and (\ref{assump:PI}), \cite{lu_jiang_ding_2023} demonstrated nonparametric identification of PCEs. Given that knowledge of this correlation parameter is a strong assumption, and that results can be sensitive to its choice, many authors have taken a parametric modeling approach to estimate this correlation. Additionally, in most of these works, principal ignorability is not assumed, showing that the parametric modeling assumptions reduce the number of structural assumptions that need to be made. This is the focus of our work, as we describe the extent to which inference in these settings is possible in the following sections. While we focus primarily on continuous intermediate variables, additional discussion extending our results to binary intermediates is included in Appendix C, and, when appropriate, we will highlight differences between these two cases.

\section{On the possibility of identifying \texorpdfstring{$\rho$}{rho} \label{sec: identifying rho}}

\subsection{A simple parametric model}
Most approaches that jointly model the observed outcome and potential intermediates are Bayesian, which will be the focus of our manuscript, however, the results have similar implications for frequentist approaches aiming to estimate $\rho$, such as those seen in \cite{bartolucci_grilli_2011}. Bayesian approaches are popular for this choice (\cite{jin_rubin_2008}, \cite{schwartz_li_mealli_2011}, \cite{kim2019bayesian}, \cite{antonelli_mealli_beck_mattei}) as they naturally handle missing data and can propagate uncertainty from missing data imputations into the final causal estimands. This general approach consists of two parts: one is the specification of a model for the observed outcome and the other is the specification of a model for the joint values of the two potential intermediates, i.e., the principal strata.

For the purpose of illustration, we first focus on a simple, linear parametric model for the outcome as well as the principal strata. There are two main reasons for considering this simple setting. For one, this is arguably the simplest modeling strategy that is making the strongest parametric assumptions, and therefore one would assume is the most likely scenario to obtain identification. However, we show that even under such strong parametric assumptions, point identification is not possible in this setting, which has implications for other approaches in the literature that utilize more complicated modeling strategies. Second, this scenario leads to simple, closed-form expressions for many partial identification regions of interest that help build intuition for the identification issues present in these models. Importantly, we extend these ideas to more complex semiparametric and nonparametric models in Section \ref{sec:extensions}. We focus on a model given by
\begin{equation*}
    \begin{split}
        E(Y_{i} \mid \bd{U}_i,\bX_i,T_i=t)&=\bbeta_{t}^T\bd{U}_i+\lambda_t+\bgamma^T\bX_i,\\
        E(\bd{U}_i \mid \bX_i)&= (\phi_0+\balpha^T\bX_i,\phi_1+\balpha^T\bX_i)^T,
    \end{split}
    \label{model:SimpleBayesianModel}
\end{equation*}
where $\bbeta_t=(\beta_{t0},\beta_{t1})$ for $t=0,1$. Our focus is on understanding when these model parameters are identified, and the extent to which they are partially identified when point identification is not possible.

As mentioned earlier, we primarily focus on settings with a continuous intermediate variable, though we present results for the binary intermediate variable setting in Appendix C. As we discuss in subsequent sections as well as in Appendix C, identification is more challenging for the continuous intermediate setting, and is therefore the focus of our manuscript. We assume Gaussian errors in the outcome model and a joint Gaussian distribution for the principal strata. We define $\mu_{yit}$ to be the conditional mean of the potential outcomes, $E[Y_{it}\mid\bd{U}_i,\bX_i,T_i=t]$, and $\bmu_{si}=(\mu_{si0},\mu_{si1})$ to be the conditional mean of principal strata. Thus, a full data generating model can be written as:
\begin{equation}
    \begin{split}
        Y_{i}\mid \bd{U}_i,\bX_i,T_i=t&\sim N(\mu_{yit},\sigma_y^2),\qquad t=0,1.\\
        \bd{U}_i \mid \bX_i &\sim N(\bmu_{si},\bSigma_{s}),
    \end{split}
    \label{model:ContinuousWithCovariates}
\end{equation}
where \begin{equation*}
    \bSigma_{s}=\begin{pmatrix}
\sia^2& \rho\sia\sib \\
 \rho\sia\sib  &  \sib^2
\end{pmatrix}.
\end{equation*}

Of particular interest is $\rho$, the association parameter of the principal strata, which plays an important role in estimating PCEs. To the best of our knowledge, there is no literature formally studying whether it is possible to identify this association parameter, which we do in the following section. 

Model (\ref{model:ContinuousWithCovariates}) becomes a Bayesian model when we incorporate prior distributions for all parameters. Note that identification is inherently more nuanced for Bayesian models as informative prior distributions can lead to weak identification of model parameters. For a general discussion on Bayesian inference in such scenarios, we point readers to \cite{gustafson_2010}. We do not consider that scenario here as we assume flat or uninformative priors for all parameters. When we discuss identification throughout, we are simply referring to scenarios where the likelihood is maximized at the true parameter values. Similarly, when referring to partial identification, we are focused on situations where the likelihood is maximized by a region of values that are equally supported by the observed data, which can not be distinguished even with an infinite sample size.

\subsection{Identification of the association parameter of the principal strata}
\label{subec:IdentifyRho}
For a continuous intermediate, estimating average PCEs always requires estimating the principal density, which involves the association parameters. Identifying the association parameters is generally infeasible due to the lack of joint observations of principal strata. Some researchers assume these are known \citep{jin_rubin_2008}, while others have tried to estimate them using similar models to those in (\ref{model:ContinuousWithCovariates}) (\cite{schwartz_li_mealli_2011}, \cite{bartolucci_grilli_2011}). The idea is that although $S(0)$ and $S(1)$ are never jointly observed, their association is implicitly included in the outcome model. The following result provides theoretical justification to this strategy:
\begin{theorem}
    Suppose that $\btheta=(\bbeta_0,\bbeta_1,\lambda_0,\lambda_1,\sigmay^2,\bgamma,\balpha,\sigmasa,\sigmasb,\phi_0,\phi_1)$ are known and principal ignorability fails: that is, at least one of $\beta_{01}$ or $\beta_{10}$ is nonzero. Under model (\ref{model:ContinuousWithCovariates}), the posterior mode of $\rho \mid \mathcal{D}$ is consistent.
    \label{thm:IdentifyRho}
\end{theorem}
Additionally, in Appendix C we show an analogous result holds in situations with a binary intermediate. Theorem \ref{thm:IdentifyRho} shows when principal ignorability is violated, then identification of $\rho$ is indeed possible, though this result should be interpreted with caution. For one, it required the outcome model to be correctly specified with \textit{known} parameters. We will see in the subsequent sections that the model parameters are only partially identified in many settings, and the implications of Theorem \ref{thm:IdentifyRho} would only be useful for estimating $\rho$ if we can consistently estimate these regression parameters. Another issue is that this result relies on principal ignorability being violated. One can show that the posterior of $\rho$ reduces to the prior distribution for $\rho$ when principal ignorability holds. This can also be seen in Theorem \ref{thm:ApproximateVariance}, where we show the asymptotic variance of the posterior distribution of $\rho$.
\begin{theorem}[Asymptotic approximation of posterior variance] \label{thm:ApproximateVariance} 
    Suppose that $\btheta$ are known and principal ignorability fails: that is, at least one of $\beta_{01}$ or $\beta_{10}$ is nonzero. Under model (\ref{model:ContinuousWithCovariates}), an asymptotic approximation of the posterior variance of $\rho \mid \mathcal{D}$ is
    \begin{equation*}
        \begin{split}
        Var(\rho\mid \Dc)&\approx n^{-1} \bigg[\Bar{T}\bybat^2\sigmasat^2 \bigg\{\frac{2\rhot^2\bybat^2\sigmasat^2}{(\zeta_1^*-{\psi_1^*}^2)^2}+\frac{1}{\zeta_1^*-{\psi_1^*}^2} \bigg\}+\\
        &\qquad (1-\Bar{T}) \byabt^2\sigmasbt^2 \bigg\{\frac{2\rhot^2\byabt^2\sigmasbt^2}{(\zeta_0^*-{\psi_0^*}^2)^2}+\frac{1}{\zeta_0^*-{\psi_0^*}^2}  \bigg\} \bigg]^{-1},
    \end{split}
    \label{eq: posterior variance}
    \end{equation*}
    where $\Bar{T}=\suml{i=1}^nT_i/n$, $\zeta^*_0-{\psi^{*}_0}^2=\sigmayt^2+(1-\rhot^2)\byabt^2\sigmasbt^2$, and $\zeta_1^*-{\psi^*_1}^2=\sigmayt^2+(1-\rhot^2)\bybat^2\sigmasat^2$.
\end{theorem}

Note that asterisks represent the true parameter values. We can see that $\byabt$ and $\bybat$ appear in the denominator of the approximation for the posterior variance. These two parameters together indicate the level of violations of principal ignorability, and we see that when both these parameters are close to zero, the posterior variance will be very large. Theorems \ref{thm:IdentifyRho} and \ref{thm:ApproximateVariance} together show that either principal ignorability holding, or even very slight violations of principal ignorability, will not be sufficient for accurately estimating the unknown correlation $\rho$. Overall, these results show that identification of $\rho$ is possible, but only if the outcome model is correctly specified and its parameters are able to be consistently estimated, and if principal ignorability fails. In the following section, we study this first issue of identifying the parameters of the outcome model in more detail when principal ignorability is violated.

\section{Partial identification when principal ignorability fails\label{sec: partial identification when PI fails}}

\subsection{Partial identification and principal causal effects}
\label{subsec: partial identification and principal causal effects}

Partial identification commonly arises in causal inference problems due to the inherent missingness of potential outcomes. Identification can commonly be achieved under stronger assumptions, but these are not always plausible and incorrectly assuming them can lead to bias when estimating causal effects. In Section \ref{sec:NotationExisting}, we reviewed common assumptions for binary and continuous intermediate variables that can lead to identification of principal causal effects. Within the context of continuous intermediates, nonparametric identification has been established under both a principal ignorability assumption, as well as a distributional assumption on the potential intermediate that assumes $\rho$ is known \citep{lu_jiang_ding_2023, zhang_yang_2025}. Other papers have estimated principal causal effects in this setting without assuming principal ignorability \citep{kim2019bayesian,kim2024bayesian}, and at times also not assuming knowledge of $\rho$ \citep{schwartz_li_mealli_2011, bartolucci_grilli_2011}. In this section, we bridge this gap by providing information about what is possible when principal ignorability is not assumed, and whether then $\rho$ can be subsequently estimated as suggested by the existing literature and the results in Section \ref{sec: identifying rho}. We do so in the context of the simple, linear model presented in (\ref{model:ContinuousWithCovariates}), but as we see in Section \ref{sec:extensions}, these results have implications for a much broader class of models. 

We now show that the parameters in Model (\ref{model:ContinuousWithCovariates}), which does not assume principal ignorability, are generally not identifiable even in this simple, parametric setting when $\rho$ is also treated as known. Of particular interest are the parameters dictating violations of principal ignorability, given by $(\byab,\byba)$. To see this, we can look at the distribution of the observed data, conditional on $T_i=t$ and $\bX_i=\bx$, which is given by
\begin{equation}
    (Y_i,S_i)\mid T_i=t,\bX_i=\bx\sim N \Bigg((\mu_{yi},\mu_{si}),\begin{pmatrix}
    \zeta_t& \psi_t\sigma_{st} \\
    \psi_t\sigma_{st} &  \sigma^2_{st},
    \end{pmatrix} \Bigg),
    \label{eq:MarginalDist}
\end{equation}
where $\mu_{yi}=\lambda_t+\bbeta_t^T(\phi_0+\balpha^T\bx,\phi_1+\balpha^T\bx) +\bgamma^T\bx$, $\mu_{si}=\phi_t+\balpha^T\bx$, $\zeta_{t}=\sigma_y^2+\bbeta_t^T\bSigma_{s}\bbeta_t$, and $\psi_t=\sigma_{st}\beta_{tt}+\rho\sigma_{s,1-t}\beta_{t,1-t}$ for $t = 0,1$. Clearly the parameters of the observed data distribution are identifiable from the observed data, and given $\bX_i=\bx$, there are 5 identifiable parameters for each of $t=0,1$ leading to 10 identifiable parameters. Excluding the parameters corresponding to $\bX$ given by $(\bgamma,\balpha)$, Model (\ref{model:ContinuousWithCovariates}) has 11 parameters, showing that there are fewer known equations than unknown variables. This leads to partial identification of certain parameters, despite being the simplest, parametric model possible that allows for violations of principal ignorability. Naturally, this also implies that the parameters of more complex semi- or nonparametric models, such as those seen in the literature \citep{bartolucci_grilli_2011, schwartz_li_mealli_2011}, are at best partially identified in this setting unless additional constraints are imposed on the model parameters.

In the presence of such partial identification, we are left with two options. One can proceed with inference where the parameters are only partially identified, and we detail the size of such partial identification regions in Section \ref{subsec:PartialIDregions}. Alternatively, additional assumptions can be made, which allows for point identification of the model parameters, and potentially $\rho$ as well, which we detail in Section \ref{sec: strong identification assumption}. We also provide weaker, alternative assumptions in Section \ref{sec: alternative assumptions} that do not obtain point identification, but can drastically reduce the widths of partial identification regions leading to more informative inference. 

\subsection{Partial identification regions when principal ignorability fails}
\label{subsec:PartialIDregions}

As mentioned previously, we focus on partial identification regions for parameters of interest without considering prior distributions for these parameters. Partial identification regions are equally supported by the observed data in the sense that they yield the same likelihood, but the posterior is not necessarily flat in these regions, due to the influence of the prior distribution. Clearly, the prior distribution affects the posterior distribution, and if strong prior knowledge is available, then this should be incorporated. We assume that such knowledge is not available and that flat priors are placed on unidentified parameters. Using $(\byab,\byba)$ and the identifiable parameters in (\ref{eq:MarginalDist}), the PCEs of interest can be written as 
\begin{equation*}
    \begin{split}
        PCE(\bd{u})&=(\byba-\frac{\psi_0}{\sia}+\frac{\sib}{\sia}\rho\byab)(s_0-\phi_0)+\\
        &\qquad (\frac{\psi_1}{\sib}-\frac{\sia}{\sib}\rho\byba-\byab)(s_1-\phi_1)+(\mu'_{y1}-\mu'_{y0}),
    \end{split}
    \label{eq:PCE in terms of beta}
\end{equation*}
where $\mu'_{yt}=\lambda_t+\bbeta_t^T(\phi_0,\phi_1)$ for $t = 0, 1$, and $(\byab,\byba)$ are subject to the constraints imposed by (\ref{model:ContinuousWithCovariates}) and (\ref{eq:MarginalDist}). Note that $\mu'_{yt}$ is identifiable and corresponds to the mean of potential outcomes, excluding the effects of covariates. Details about their identification are provided in Appendix A. The PCEs depend on the principal strata $\bd{u} = (s_0, s_1)$ and on $(\byab,\byba)$, and therefore we aim to study the partial identification region of $(\byab,\byba)$. For simplicity of the expressions for the partial identification regions of $(\byab,\byba)$, we exclude covariates from the theoretical results regarding partial identification regions. Note that in the Appendix we show that to incorporate covariates, we only need to condition additionally on $\bX$ for all conditional variances or correlations found in the partial identification regions.

\begin{proposition}
    \label{prop: PIR for byab and byba without constraints}Suppose that $\Var\{Y(1)\mid S(1)\}\geq\Var\{Y(0)\mid S(0)\}$.
    Then the partial identification region of $\byab$ is \begin{equation*}
        \left[-\sqrt{\frac{\Var\{Y(0)\mid S(0)\}}{(1-\rho^2)\Var\{S(1)\}}},\sqrt{\frac{\Var\{Y(0)\mid S(0)\}}{(1-\rho^2)\Var\{S(1)\}}} \right],
        \label{eq: PIR for byab without constraints}
    \end{equation*} 
    and the partial identification region of $\byba$ is 
    \begin{equation*}
        \begin{split}
            & \left[-\sqrt{\frac{\Var\{Y(1)\mid S(1)\}}{(1-\rho^2)\Var\{S(0)\}}},-\sqrt{\frac{\Var\{Y(1)\mid S(1)\}-\Var\{Y(0)\mid S(0)\}}{(1-\rho^2)\Var\{S(0)\}}} \right] \cup\\&\qquad \left[\sqrt{\frac{\Var\{Y(1)\mid S(1)\}-\Var\{Y(0)\mid S(0)\}}{(1-\rho^2)\Var\{S(0)\}}},\sqrt{\frac{\Var\{Y(1)\mid S(1)\}}{(1-\rho^2)\Var\{S(0)\}}} \right],
            \label{eq: PIR for byba without constraints}
        \end{split}
    \end{equation*}
    where $\byab$ and $\byba$ also satisfy the constraint
    \begin{equation}
        \Var\{S(0)\}\byba^2-\Var\{S(1)\}\byab^2=\frac{\Var\{Y(1)\mid S(1)\}-\Var\{Y(0)\mid S(0)\}}{1-\rho^2}.
        \label{eq: beta constraint}
    \end{equation}
\end{proposition}

Note that if $\Var\{Y(0)\mid S(0)\}\geq\Var\{Y(1)\mid S(1)\}$, one can simply exchange 0 and 1 in Proposition \ref{prop: PIR for byab and byba without constraints} to obtain the corresponding partial identification region. Also note that all terms in the partial identification region are identifiable since we are assuming for now that $\rho$ is known. An interesting consequence of equation (\ref{eq: beta constraint}) is that if the right-hand side is non-zero, then at least one of $\byab$ and $\byba$ must be nonzero, indicating principal ignorability is necessarily violated. Proposition \ref{prop: PIR for byab and byba without constraints} shows that both the sign and magnitude of ($\byab,\byba$) are not identified, and that these regions can be rather large, which will lead to large partial identification regions for the PCEs of interest. Further, we can see an interplay between the size of the partial identification region for the two parameters. When the partial identification region for $\beta_{01}$ is larger due to a large value of $\Var(Y(0) \mid  S(0))$, this leads to a reduction in the width of the region for $\beta_{10}$, and vice-versa. This interplay also effectively guarantees that the partial identification region for at least one of the parameters (if not both) is large, and therefore principal causal effects will have wide partial identification regions as well. Interestingly, in Appendix C we show that for binary intermediates, the magnitudes of both $\beta_{01}$ and $\beta_{10}$ are in fact identifiable, though their signs are not. This shows that the continuous intermediate setting is inherently more challenging than the binary intermediate setting. 

\subsection{Assumptions for identification when principal ignorability is violated}
\label{sec: strong identification assumption}

The previous sections cast doubt on the ability to perform inference when principal ignorability is violated, and additionally question whether the results from Section \ref{subec:IdentifyRho} are useful in practice, because they showed that $\rho$ can only be identified under known outcome model parameters. Here, we discuss different assumptions that can not only lead to the identification of PCEs given known $\rho$, but also to the identification of $\rho$ itself. As discussed earlier, there are 10 identified parameters in (\ref{eq:MarginalDist}), while Model (\ref{model:ContinuousWithCovariates}) has 11 unknown parameters if $\rho$ is known, or 12 unknown parameters otherwise. The easiest way to make these parameters point identifiable is to add constraints on the parameters of Model (\ref{model:ContinuousWithCovariates}) that reduce the number of unknown parameters. For instance, principal ignorability can be viewed as imposing the constraint $\byab = \byba = 0$, which leads to the identification of all parameters in Model (\ref{model:ContinuousWithCovariates}). Despite there being two constraints imposed, this particular assumption does not identify $\rho$, because Theorem \ref{thm:IdentifyRho} showed that principal ignorability must be violated to obtain identification of $\rho$. 

Alternative assumptions beyond the principal ignorability assumption have been proposed in the literature, which we can adopt to obtain point identification. One such modeling assumption explored in recent work \citep{schwartz_li_mealli_2011, bartolucci_grilli_2011} in our setting amounts to setting $\byaa=\byba,\byab=0$. This assumption treats $E\{Y(0)\mid \bd{U}, \bX\}$ as a baseline characteristic that is shared between the treatment and control groups \citep{schwartz_li_mealli_2011}. This assumption identifies $\rho$ only up to its sign, however, it is typically reasonable to assume $\rho > 0$, which enables full identification of $\rho$. Alternatively, one could make the assumption that both $\lambda_0 = \lambda_1$ and $\byab = 0$, which amounts to assuming there is no direct effect of treatment on the outcome and that the unobserved intermediate does not affect the outcome in the control group. Note that both of these possibilities impose two constraints on the parameter space, which allows us to identify both the outcome model parameters and the unknown correlation $\rho$. If $\rho$ is treated as known, or as a sensitivity parameter to be varied, then only one constraint is needed to obtain identification of the outcome model parameters. It is important to emphasize that while these constraints can be placed to obtain identification, they should only be made if they are deemed plausible in the application to which they are utilized. If one can not make such strong assumptions, then one should continue with partial identification. To this end, we now shift focus to incorporating more plausible assumptions that do not obtain point identification, but can reduce the widths of partial identification regions significantly when such strong identifying assumptions are not plausible.

\section{Weaker alternative assumptions of principal ignorability}
\label{sec: alternative assumptions}

Strong assumptions like principal ignorability or the assumptions discussed in Section \ref{sec: strong identification assumption} can lead to point identification of principal causal effects, but they are untestable and may be overly restrictive in many settings. Without them, however, partial identification regions such as those seen in Proposition \ref{prop: PIR for byab and byba without constraints} can be overly wide and uninformative, even in contexts with simple, parametric models. This is expected to persist or be exacerbated in the context of more flexible outcome models. To allow for violations of principal ignorability and reduce the widths of partial identification regions, we propose two weaker assumptions that can be reasoned about in any particular application. The first assumption addresses the issue of unidentified signs in the effects of the unobserved intermediate. 
 
\begin{assumption}[\label{assump: same sign}Same Sign]
    Assume 
    \begin{equation*}
        \text{sign} \Bigg\{\dfrac{\partial E(Y\mid T,\bd{U},\bX)}{\partial S(T)}\Bigg\} =\text{sign} \Bigg\{\dfrac{\partial E(Y\mid T,\bd{U},\bX)}{\partial S(1-T)}\Bigg\}. 
    \end{equation*} 
\end{assumption}

Assumption \ref{assump: same sign} states that, given the covariates and treatment, the observed and unobserved intermediates should have effects in the same direction on the outcome. This is plausible in many applications where the effect of the intermediate on the outcome can only plausibly be in one direction. In the case of Model (\ref{model:ContinuousWithCovariates}), this simplifies to assuming that $\beta_{01}$ and $\beta_{00}$ have the same sign, and that $\beta_{10}$ and $\beta_{11}$ have the same sign. Clearly this assumption rules out a significant portion of the partial identification region found in Proposition \ref{prop: PIR for byab and byba without constraints}, which is shown in the following result.
Under Assumption \ref{assump: same sign}, only the magnitudes of ($\byab,\byba$) are unidentified.

\begin{proposition}\label{prop: PIR for byab and byba assuming same signs}
    Suppose that Assumption \ref{assump: same sign} holds and $\Var\{Y(1)\mid S(1)\}\geq\Var\{Y(0)\mid S(0)\}$. Assume $\byaa\geq0$ and $\bybb\geq0$. The partial identification region for $\byab$ is 
    \begin{equation*}
        \Bigg[0,\sqrt{\frac{\Var\{Y(0)\mid S(0)\}}{(1-\rho^2)\Var\{S(1)\}}} \Bigg],
    \end{equation*}
    and the partial identification region for $\byba$ is
    \begin{equation*}
        \Bigg[\sqrt{\frac{\Var\{Y(1)\mid S(1)\}-\Var\{Y(0)\mid S(0)\}}{(1-\rho^2)\Var\{S(0)\}}},\sqrt{\frac{\Var\{Y(1)\mid S(1)\}}{(1-\rho^2)\Var\{S(0)\}}} \Bigg],
    \end{equation*}
    where $\byab$ and $\byba$ also satisfy (\ref{eq: beta constraint}).
\end{proposition}

For $\byaa$ and $\bybb$ with signs different from those assumed in Proposition \ref{prop: PIR for byab and byba assuming same signs}, the corresponding partial identification regions for $\byab$ and $\byba$ can be obtained by flipping the signs to match those of $\byaa$ and $\bybb$, respectively. This result also requires us to know the sign of $\beta_{tt}$ for $t=0,1$. If this is not known a priori, then one can regress the observed $Y$ on the observed $S$ and $\bX$, given $T=t$, to determine the sign for $\beta_{tt}$. This approach is not guaranteed to correctly identify the sign, but will only fail to do so in certain extreme situations, such as when $\rho$ is close to $-1$. This cuts the partial identification region in half, providing large benefits in the uncertainty in the resulting causal estimates. While these benefits may be large, it is not possible in many situations to know this assumption is true a priori. Now we present a second, weaker assumption, which should hold true in most applications. This second assumption effectively assumes that the unobserved intermediate, once we condition on the observed one, has a smaller impact on the outcome than the observed one had. This can be seen as a weakened version of principal ignorability as the principal ignorability assumption assumes that there is no effect of the unobserved intermediate, once we condition on the observed one. This is formalized in the following assumption.
\begin{assumption} [\label{assump: dominant ob effect}Dominant Observed Effect]
    Assume 
    \begin{equation*}
        R^2_{Y\sim S(1-T) \mid S(T),T,\bX}\leq R^2_{Y\sim S(T) \mid T,\bX}.
    \end{equation*}
\end{assumption}

Assumption \ref{assump: dominant ob effect} implies that the signal-to-noise ratio (SNR) of the unobserved intermediate, conditional on the observed intermediate, is no greater than that of the observed intermediate. This assumption relaxes principal ignorability by allowing the unobserved intermediate to affect the outcome, but precludes large effects of the unobserved intermediate. Although Assumption \ref{assump: dominant ob effect} is formulated in terms of Pearson's $R^2$, one could also formulate the assumption in terms of nonparametric partial $R^2$ values that would apply in broader contexts, as we see in Section \ref{sec:extensions}. This assumption reduces the width of the partial identification region substantially as shown in the following result. 

\begin{proposition}\label{prop: PIR for byab and byba assuming dominant ob effect 1}
    Suppose that Assumption \ref{assump: dominant ob effect} holds, $\Var\{Y(1)\mid S(1)\}\geq\Var\{Y(0)\mid S(0)\}$, $\frac{(\Var\{Y(0)\mid S(0)\})^2}{\Var(Y(0))}\geq\frac{(\Var\{Y(1)\mid S(1)\})^2}{\Var(Y(1))}$, and $\byaa>0$. The partial identification region for $\byab$ is 
    \begin{equation*}
        \Bigg[-\sqrt{\frac{\Var\{Y(0)\mid S(0)\}\Cor^2(Y(0),S(0))}{(1-\rho^2)\Var\{S(1)\}}},\sqrt{\frac{\Var\{Y(0)\mid S(0)\}\Cor^2(Y(0),S(0))}{(1-\rho^2)\Var\{S(1)\}}} \Bigg],
    \end{equation*}
    and the partial identification region for $\byba$ is
    {\footnotesize
    \begin{equation*}
        \begin{split}
            & \Bigg[-\sqrt{\frac{\Var\{Y(1)\mid S(1)\}-\Var\{Y(0)\mid S(0)\}/\Var(Y(0))}{(1-\rho^2)\Var\{S(0)\}}},-\sqrt{\frac{\Var\{Y(1)\mid S(1)\}-\Var\{Y(0)\mid S(0)\}}{(1-\rho^2)\Var\{S(0)\}}} \Bigg]\cup\\
            &\qquad \Bigg[\sqrt{\frac{\Var\{Y(1)\mid S(1)\}-\Var\{Y(0)\mid S(0)\}}{(1-\rho^2)\Var\{S(0)\}}},\sqrt{\frac{\Var\{Y(1)\mid S(1)\}-\Var\{Y(0)\mid S(0)\}/\Var(Y(0))}{(1-\rho^2)\Var\{S(0)\}}} \Bigg],
        \end{split}
    \end{equation*}
    }
    where $\byab$ and $\byba$ also satisfy (\ref{eq: beta constraint}).
\end{proposition}

Note that the result depends on $\Var\{Y(1)\mid S(1)\}\geq\Var\{Y(0)\mid S(0)\}$, ${(\Var\{Y(0)\mid S(0)\})^2}\break /{\Var(Y(0))}\geq{(\Var\{Y(1)\mid S(1)\})^2}/{\Var(Y(1))}$, and $\byaa>0$ all being true. We should emphasize that similar bounds can be derived under any combination of directions for these inequalities, but we focus on one such choice for simplicity. By comparing the results in Propositions \ref{prop: PIR for byab and byba without constraints} and \ref{prop: PIR for byab and byba assuming dominant ob effect 1},
it is clear that Assumption \ref{assump: dominant ob effect} reduces the length of the partial identification region substantially. For $\byab$ the size of the region is reduced by a factor of $| \Cor(Y(0),S(0))| $, which reflects the strength of the association between the outcome and the observed intermediate within the control $(T=0)$ group. When $| \Cor(Y(0),S(0))|$ is very small, the benefits brought by Assumption \ref{assump: dominant ob effect} are large since it prevents assigning overly large values to the effect of the missing intermediate. 


One interesting aspect of this result, which we show in the proof of Proposition \ref{prop: PIR for byab and byba assuming dominant ob effect 1} in the Appendix, is that Assumption \ref{assump: dominant ob effect} is equivalent to assuming $$\dfrac{[\Var\{Y\mid S,\bX,T=t\}]^2}{\Var(Y\mid \bX,T=t)}\leq\Var\{Y(t)\mid S(t),S(1-t),\bX,T=t\},$$
which shows that the assumption is equivalent to a lower bound on the residual variance, $\sigma_y^2$. One important implication of this is the ability to apply such a constraint within our Bayesian modeling framework. One can first estimate this lower bound using the observed data, and can then place a truncated prior distribution on $\sigma_y^2$ that guarantees the bound holds and therefore Assumption \ref{assump: dominant ob effect} also holds. This is also important once semiparametric or nonparametric outcome models are used, because it can be more difficult in those settings to enforce restrictions on the estimated regression functions, but a truncated prior on the residual variance is similarly straightforward in that setting.

\section{Simulations and applications}
\label{Sec: Simulations and applications}
\subsection{Simulations}
\label{subsec: simulations}
Here, we examine the theoretical results obtained in the previous sections as well as the practical performance of our weakened identification assumptions for estimating principal causal effects. In the first set of simulation studies, we treat $\rho$ as known, and focus on partial identification of principal causal effects of interest under varying assumptions. Then, we explore the identification of $\rho$ under varying assumptions and sample sizes. Throughout, we do not include covariates $\bX$, but similar results would be obtained if covariates were additionally adjusted for. Under both scenarios, Gibbs sampling is used, with a single MCMC chain run for a total of 25,000 iterations, where the first 5,000 are burn-in, and a thinning interval of 30 is applied. The details of the Gibbs sampler are included in Appendix B.

\subsubsection{Performance under a known \texorpdfstring{$\rho$}{rho}}
In this scenario, the data-generating model is based on the analysis of \cite{bartolucci_grilli_2011} and the dataset from the Lipid Research Clinics Coronary Primary Prevention Trial (LRC-CPPT) (\cite{efron_feldman_1991}). The parameter values are based on the estimates from previous work with small adjustments in order to match Model (\ref{model:ContinuousWithCovariates}). Specifically, the data are generated from the following data generating process:
\begin{equation}
    \begin{split}
        Y_i(0)\mid \bd{U}_i&\sim N(-0.5+11.5S(0),14^2),\\
        Y_i(1)\mid \bd{U}_i&\sim N(-0.5+11.5S(0)+96S(1),14^2),\\
        \bd{U}_i&\sim N((0.89,0.70)^T,\boldsymbol{\Sigma}_S),
    \end{split}
    \label{eq:Setting1}
\end{equation}
where \[
\boldsymbol{\Sigma}_S=\begin{pmatrix}
    0.25^2 & 0.75*0.25*0.25\\
    0.75*0.25*0.25 & 0.25^2
\end{pmatrix}.
\]

The parameter values here are based on previous studies (\cite{bartolucci_grilli_2011}, \cite{schwartz_li_mealli_2011}), both of which conducted analyses using the same outcome models. We assume that $\rho$ is known and focus on partial identification of the PCEs under violations of principal ignorability. One can verify that Model (\ref{eq:Setting1}) satisfies Assumption \ref{assump: dominant ob effect}, but only partially satisfies Assumption \ref{assump: same sign} because $\beta_{01} = 0$ and therefore does not have the same sign as $\beta_{00}$, which is positive. We generate 500 simulated datasets from Model (\ref{eq:Setting1}) with $n=300$. For each simulated dataset, we run three MCMC chains: the first is without either Assumption \ref{assump: same sign} or \ref{assump: dominant ob effect}, the second incorporates Assumption \ref{assump: dominant ob effect}, and the third incorporates Assumption \ref{assump: same sign}. Note that we only partially apply Assumption \ref{assump: same sign} by applying it to the signs of $(\beta_{10}, \beta_{11})$, which are equal. For the model without constraints or assumptions, we consider noninformative conjugate priors as follows:
\begin{equation*}
    \begin{split}
        &\bbeta_t\sim N(0,10^5\bd{I}_2),\ 
        \lambda_t\sim N(0,10^5),\ 
        \phi_t\sim N(0,10^5),\\
        &\sigma_y^2\sim IG(10^{-3},10^{-3}),\ 
        \sigma_{s}^2\sim IG(10^{-3},10^{-3}),
    \end{split}
\end{equation*}
where $IG(a,b)$ represents the inverse gamma distribution and $\bd{I}_2$ is a $2\times2$ identity matrix. We have noticed empirically that partial identification can lead to inefficient MCMC sampling with parameters getting trapped at extreme values. In this case, the residual variance $\sigma_y^2$ becomes trapped at values close to 0, so a lower bound of $0.05\min_{t=0,1}\Var(Y\mid T=t)$ is used throughout when updating $\sigma_y^2$. This improves the poor mixing, but the issue still persists to a lesser degree. For MCMC chains under Assumption \ref{assump: dominant ob effect}, since Assumption \ref{assump: dominant ob effect} is equivalent to adding a lower bound for $\sigma_y^2$, then a truncated inverse gamma distribution for $\sigma_y^2$ is utilized, where a lower bound for this truncation is the empirical estimate of $0.9\min_{t=0,1}\left((\Var[Y\mid S,T=t])^2/\Var(Y\mid T=t)\right)$. For MCMC chains that partially incorporate Assumption \ref{assump: same sign}, a truncated normal prior for $\bbeta_1$ is used to ensure that $(\byba>0,\bybb>0)$.

We target PCEs defined on ($S(0),S(1)$) at the combinations of the quartiles of the intermediate values, as used by \cite{jin_rubin_2008,bartolucci_grilli_2011,schwartz_li_mealli_2011}. Since the true values of the PCEs in this scenario depend only on $S(1)$, we fix $S(0)$ at its median value. The simulation results are shown in Table \ref{table:CredibleIntervals}. All empirical coverage rates (ECRs) are above 95\%, indicating that valid credible intervals (CIs) are obtained under any of the proposed assumptions. The average widths of CIs were compared, and the proposed assumptions significantly reduced average widths. Particularly, for $S(1)$ at the first and third quartiles, Assumption \ref{assump: dominant ob effect} provided 42\% reductions in the average width and Assumption \ref{assump: same sign} provided 26\% reductions relative to not making either assumption.

\setlength{\tabcolsep}{3pt}
\begin{table}[ht]
\def~{\hphantom{0}}
\centering
\begin{tabular}{ccrrrrrrrrr}
 \multicolumn{2}{c}{} & \multicolumn{3}{c}{No constraints} & \multicolumn{3}{c}{Dominant effect} & \multicolumn{3}{c}{Same sign} \\ 
$\bd{u}$ & Truth & Mean & ECR & Width & Mean & ECR & Width & Mean & ECR & Width \\ 
$(0.89,0.18)$ & 17 & 17 & 0.99 & 43 & 16 & 1.00  & 25  & 23  & 0.97 & 32  \\ 
  $(0.89,0.35)$ & 34  & 34  & 0.99 & 13  & 34  & 1.00   & 9 & 34  & 0.98 & 12  \\ 
  $(0.89,0.52)$ & 50  & 50  & 0.99 & 43  & 51  & 1.00 & 25  & 45  & 0.98 & 32  \\ 
\end{tabular}
\caption{Comparisons of average posterior mean (Mean), empirical coverage rates (ECR) and average width of credible intervals (Width) under varying assumptions that allow for violations of principal ignorability.}
\label{table:CredibleIntervals}
\end{table}

\subsubsection{Identifying \texorpdfstring{$\rho$}{rho} under different constraints}
\label{subsec:Scenario2}

Now we use a related data generating process to further explore the ability to identify $\rho$ under varying assumptions. As discussed previously, one plausible identification assumption for $\rho$ is $\byab=0$ and $\byaa=\byba$. This guarantees the identification of $\rho$ when we also assume $\rho>0$. To investigate both identification and posterior consistency, we vary $n \in \{300, 600, 1200\}$ and we generate one dataset for each sample size. For each dataset, we run MCMC under no constraints, one constraint ($\byab=0$), and two constraints ($\byab=0$, $\byaa=\byba$). To avoid the sign issue, we also constrain $\rho$ to be in $(0.00,0.95)$, where setting $\rho<0.95$ prevents MCMC sampling from getting trapped around the extreme value $\rho=1$.

We generate data from the following data generation process:
\begin{equation*}
    \begin{split}
        Y_i(0)\mid \bd{U}_i&\sim N(0.9+1.2S(0),0.5^2),\\
        Y_i(1)\mid \bd{U}_i&\sim N(0.5+1.2S(0)+1.2S(1),0.5^2),\\
        \bd{U}_i&\sim N((0.3,0.5)^T,\boldsymbol{\Sigma}_S),
    \end{split}
    \label{eq:Setting2}
\end{equation*}
where $\boldsymbol{\Sigma}_S=\begin{pmatrix}
    1 & 0.75\\
    0.75 & 1
\end{pmatrix}$. The same noninformative conjugate priors are used and we additionally place a flat prior on the correlation between potential intermediates, $\rho\sim U(0,0.95)$.

The posterior distributions of $\rho$ across all three scenarios and sample sizes are shown in Figure \ref{fig:PoseriorPlots}. Regardless of sample size, there is a notable spike around the true value of $\rho=0.75$ when we place two constraints on the model parameters. There appear to be mild spikes close to 1 when there are no constraints or one constraint placed on the model, but this is simply due to poor mixing as the MCMC sampler gets stuck at $\rho$ values close to 1. When we alternatively assume $\rho\sim U(0,0.9)$, these spikes are reduced. Importantly, as $n$ increases, their height does not increase, whereas the height of spike under two constraints increases significantly, which points towards consistency of the correlation parameter. In Appendix D, we include more plots of the estimated posterior variance for more sample sizes, which show that asymptotically, the posterior variance is $O(n^{-1})$. 

\begin{figure}[tb]
    \centering
    \includegraphics[width=0.75\linewidth]{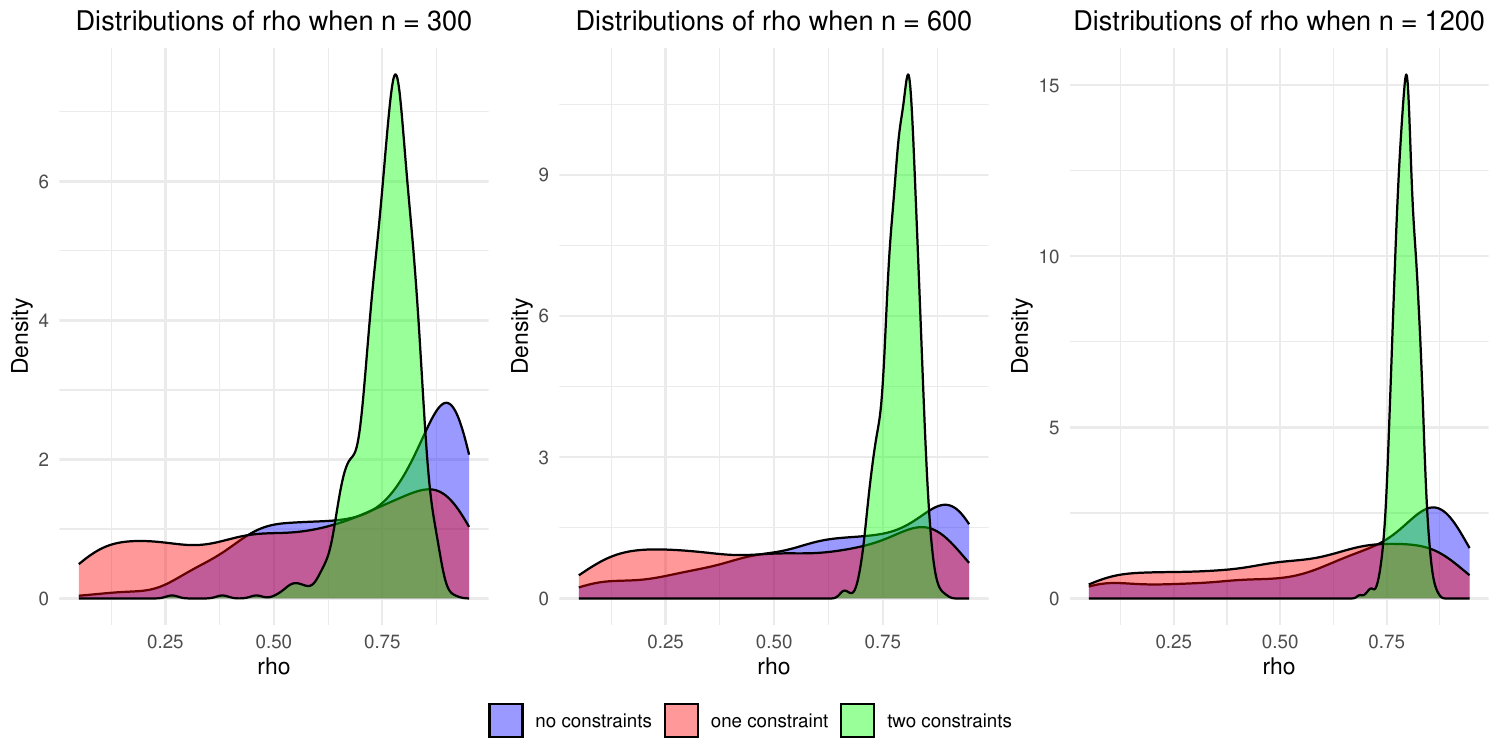}
    \caption{The posterior distributions of $\rho$ corresponding to $n=$300, 600, 1200 under three different amounts of constraints on the model parameters.}
    \label{fig:PoseriorPlots}
\end{figure}

\subsection{Analysis of ACTG trial data}
The ACTG 175 data set (\cite{hammer_katzenstein_hughes_gundacker_schooley_haubrich_henry_lederman_phair_niu_et_1996}) is available in the R package \emph{speff2trial} and was collected from a randomized clinical trial with the purpose of comparing monotherapy with zidovudine or didanosine with combination therapy with zidovudine and didanosine or zidovudine and zalcitabine in adults infected with the human immunodeficiency virus type I whose CD4 T cell counts were between 200 and 500 per cubic millimeter. This dataset was investigated by \cite{tomasz_burzykowski_geert_molenberghs_buyse_2005} and \cite{zhang_yang_2025}, with a focus on whether a short-term endpoint can serve as a valid surrogate for a long-term endpoint. Here, we explore the same question for this dataset, though we do not assume principal ignorability holds as in the analysis of \cite{zhang_yang_2025}, and we explore the extent to which our various proposed assumptions can provide informative partial identification regions. We consider the short-term endpoint, CD4 count at $20\pm5$ weeks, as the intermediate variable $S$, and the long-term endpoint, CD4 count at $96\pm5$ weeks, as the outcome $Y$. The treatment $T$ is 0 if the patient is treated with zidovudine only and is 1 if the patient is treated with zidovudine + didanosine. Covariates $\bX$ include age, weight, hemophilia, homosexual activity, history of intravenous drug use, Karnofsky score, non-zidovudine antiretroviral therapy prior to initiation of study treatment, zidovudine use in the 30 days prior to treatment initiation, race, gender, antiretroviral history, and a symptomatic indicator. We apply Model (\ref{model:ContinuousWithCovariates}) throughout with a fixed value of $\rho = 0.5$. Note that this model does not assume principal ignorability, which we show in Section \ref{sec: partial identification when PI fails} can lead to substantially wider credible intervals for the principal causal effects of interest. For this reason, we also explore incorporating both the dominant effect assumption and the same sign assumptions proposed in Section \ref{sec: alternative assumptions}, which leads to three different approaches in total as we also consider the case where no assumptions are incorporated. Throughout we consider a range of principal strata that can be broken into two types. The first looks at subjects for whom $S(0) = S(1) = s$ and we vary the value of $s$. The second set of principal strata fixes the value of $S(0) = 340$ and then varies $S(1)$ to see how the causal effect for the final endpoint depends on the effect in the surrogate endpoint. 

The results can be found in Figure \ref{fig:daPces}, which shows both the posterior mean and corresponding 95\% credible interval for all estimands of interest across the three proposed approaches. One of the most apparent findings from the results is that inference without principal ignorability and no additional assumptions (first column of Figure \ref{fig:daPces}) is less informative with wide credible intervals, but the two proposed assumptions (2nd and 3rd columns) can substantially shrink the size of the credible intervals, particularly for the same sign assumption. The results under both the dominant effect and same sign assumptions also paint a fairly clear picture about the performance of $S$ as a surrogate endpoint. When $S(0) = S(1) = s$, there are slightly positive estimates with credible intervals that generally contain zero, though the intervals under the same sign assumption are strictly positive for certain values of $s$. This suggests that if there is no effect of the treatment by the surrogate endpoint at 20 weeks, then there likely won't be a big effect of the treatment by the terminal endpoint at 96 weeks, though the results are suggestive that there is still some beneficial effect of treatment even when $S(0) = S(1)$. This indicates that it may take time for the full effect of treatment to realize, though this effect is not substantial. The results from when $S(0) \neq S(1)$ also show that the surrogate endpoint is a valid surrogate in the sense that when $S(1) - S(0)$ is larger, the treatment effect for the final endpoint described by the magnitude of $Y(1) - Y(0)$ is also large. These results are similar to those obtained in recent analyses assuming principal ignorability \citep{zhang_yang_2025}. This shows the benefit of our proposed assumptions, which are arguably weaker than principal ignorability, but are still able to provide meaningful insights on principal causal effects. If principal ignorability holds, we can still obtain similar results without greatly increasing the widths of the credible intervals, while if principal ignorability fails, our results should be more robust.

\begin{figure}[tb]
    \centering
    \includegraphics[width=0.98\linewidth]{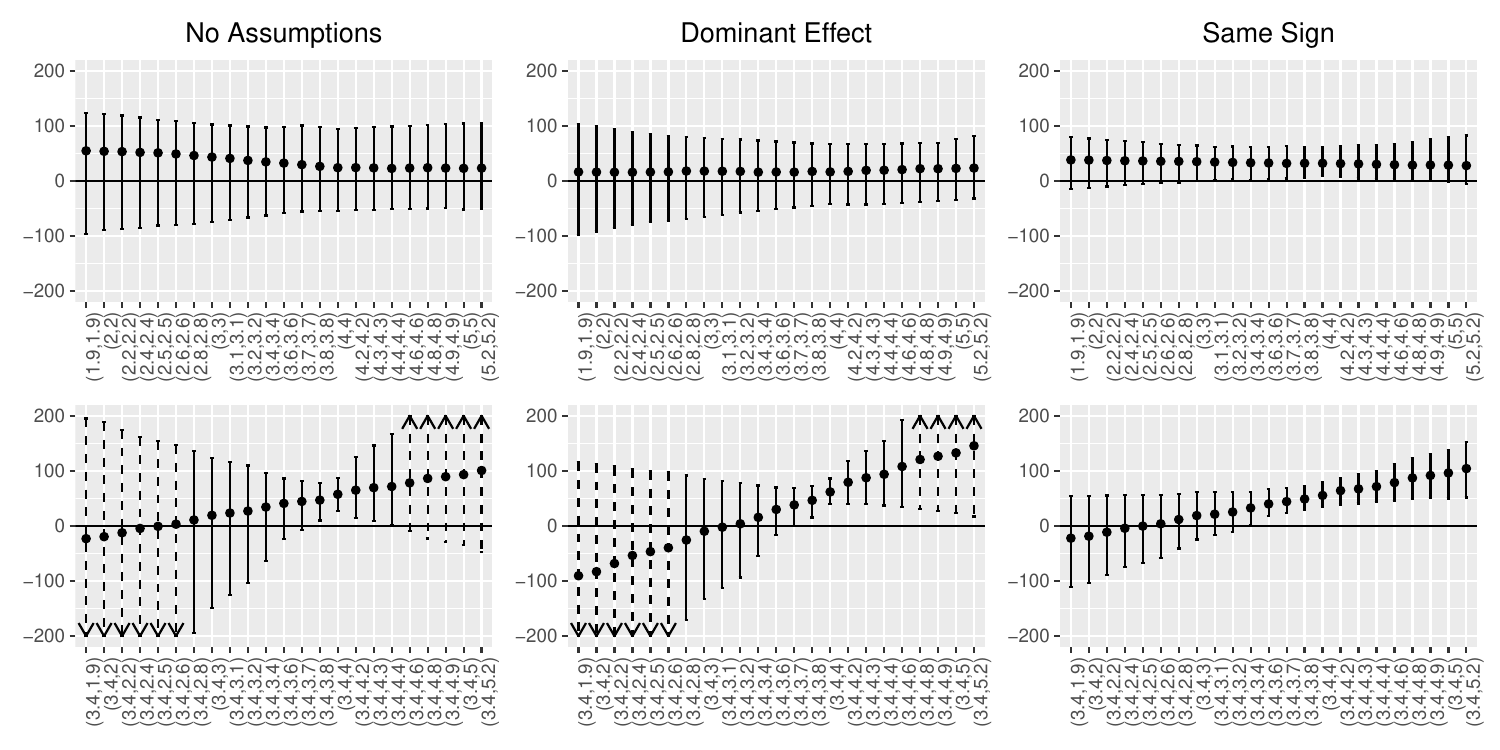}
    \caption{The first row shows PCEs and corresponding 95\% credible intervals for the principal strata defined by $S(0)=S(1) = s$ with varying $s$. The second row corresponds to the principal strata with $S(0)=340$ and increasing $S(1)$. Note that $S(0)$ and $S(1)$ are scaled by $10^2$ in the plot. From left to right, the columns correspond to different assumptions: no constraints, dominant effect, and same sign, respectively. The dashed line in the CI plots implies that the CI is truncated, and the arrow indicates which side is truncated.}
    \label{fig:daPces}
\end{figure}

\section{Extensions to Bayesian nonparametric modeling} \label{sec:extensions}

In Sections \ref{sec: partial identification when PI fails} and \ref{sec: alternative assumptions}, we discussed partial identification regions and alternative assumptions when principal ignorability fails. While many of these results were described in terms of conditional variances not unique to a specific model, they were derived under the assumption of a linear outcome model. This facilitated derivations and ease of exposition, but now we show that many similar ideas hold in more general modeling contexts, which has broader implications for the analysis of principal causal effects when principal ignorability is violated. 

We consider directly extending Model (\ref{model:ContinuousWithCovariates}) to a more flexible modeling framework that does not make the same, potentially restrictive, parametric assumptions. Specifically, we can model the outcomes as follows:
\begin{equation}
    \begin{split}
        Y_{i}\mid \bd{U}_i,\bd{X_i},T_i=t&\sim \mu_{yt}(\bd{X}_i,S_i)+\mu_{yct}(\bX_i,\bd{U}_i)+\epsilon_i,\\
    \bd{U}_i\mid \bX_i&\sim N(\bmu_{s}(\bX),\bSigma_{s}),
    \label{model:BART}
    \end{split}
\end{equation}
where $\epsilon_i\sim N(0,\sigma_y^2)$, $\epsilon\ind \bd{U}_i$, and $E(Y_{it}\mid S_i,\bd{X_i},T_i=t)=\mu_{yt}(\bd{X}_i,S_i)$.  The subscript $c$ implies that $E\{\mu_{yct}(\bX_i,\bd{U}_i)\mid S_i,\bX_i\}=0$, which means it is centered. Therefore we can view $\mu_{yct}(\bX_i,\bd{U}_i)$ as a nonparametric extension of $\beta_{t,1-t} S(1-t)$ as it controls the degree of violations of principal ignorability, and if $\mu_{yct}(\bX_i,\bd{U}_i)=0$, then principal ignorability holds. Our results here are general and hold for any nonparametric Bayesian prior placed on these functions, but either Gaussian processes or Bayesian additive regression tree (BART, \cite{chipman_george_mcculloch_2010}) priors would likely be used. The latter of which has been used for estimating principal causal effects \citep{kim2024bayesian} and is commonly used more generally within causal inference \citep{hill_2011, hahn_murray_carvalho_2020, linero2023and}.

\subsection{Partial identification related to extended sign issues}

Without principal ignorability, one aspect of partial identification concerns the unidentified signs within a parametric linear outcome model. In the nonparametric case, this concept generalizes to involve a specific class of transformations known as invariant transformations.

\begin{definition}
    A transformation $M$ of a random variable $Z$ is said to be an invariant transformation if $M(Z)$ and $Z$ follow the same distribution. 
\end{definition}

Let $\mathcal{M}(Z)$ denote the collection of all invariant transformations of a random variable $Z$. To formalize the role of invariant transformations in identifying distributions, we present the following result, which shows that principal causal effects are not identified in this setting.

\begin{theorem} \label{thm:UnidentifiedInvariant}
    For any invertible $M\in\mathcal{M}(S(1-t)\mid S(t),\boldsymbol{X})$, let $Y'(t)=\mu_{yt}\{\bX,S(t)\}+\mu_{yct}[\bX,S(t),M^{-1}\{S(1-t)\}]+\epsilon$. Then, marginally, $(Y'(t),S(t))\mid \bd{X}$ and $(Y(t),S(t))\mid \bd{X}$ follow the same distribution.
\end{theorem}
Theorem \ref{thm:UnidentifiedInvariant} shows that $\mu_{yct}(\bX, \bd{U})$ is unidentifiable, and therefore $E\{Y(t)\mid S(t),S(1-t),\boldsymbol{X}\}$ and any PCEs of interest are also unidentifiable. Because of the unobserved nature of $S(1-t)$, the flexibility of nonparameteric outcome models allows various possible ways for $S(1-t)$ to contribute to the observed outcome. These plausible outcome models can lead to completely different or even contradictory PCEs. For instance, if $(S(t), S(1-t))$ follows a bivariate Gaussian distribution, then one such invariant transformation can be constructed as $M: s \to 2m_{1-t}(s_t,\bd{x}) - s$, where $m_{1-t}(s_t,\bd{x})$ denotes $E\{S(1-t)\mid S(t)=s_t,\bd{X}=\bx\}$. Clearly, $M$ is invertible and $M^{-1} = M$. Partial identification arises if $\mu_{yct}\{\bx,s_t,2m_{1-t}(s_t,\bd{x})-s_{1-t}\} \neq \mu_{yct}(\bx,s_t,s_{1-t})$, which is typically the case. To provide intuition for this, this inequality would hold in the linear model in (\ref{model:ContinuousWithCovariates}) if $\beta_{t,1-t}$ is not equal to 0, or equivalently, when principal ignorability fails. 

Although we assume a bivariate Gaussian distribution for $(S(t),S(1-t))$ in Model (\ref{model:BART}), Theorem \ref{thm:UnidentifiedInvariant} does not rely on this assumption and holds for $(S(t),S(1-t))$ under any distribution. Theorem \ref{thm:UnidentifiedInvariant} generalizes the sign issues leading to partial identification in the linear model setting to a more general partial identification result based on a class of invariant transformations. Usually, the number of transformations in this class is infinite. However, most are artificial, resulting in complex and unusual potential outcome models $Y'(t)$, which can easily be ruled out under reasonable modeling constraints. For cases like the sign issue, one can use domain knowledge to determine the correct transformation or employ reasonably uninformative priors to rule out other, incorrect outcome models.

\subsection{Partial identification based on magnitude of unobserved intermediate effect}

Under linear outcome models, the sign of the effect of the unobserved intermediate is not identified without additional constraints, but even if the sign were known, the magnitude of the effect is also not identified. Similarly, for nonparametric outcome models, even if the invariant transformations discussed previously were not an issue, we still are not able to identify the magnitude of the effect of the unobserved intermediate variable. In Model (\ref{model:BART}), this is represented by the magnitude of the $\mu_{yct}(\bX, \bd{U})$ function. Conditional on $(S(t), \bX)$,the conditional variance of $\mu_{yct}$ is also the conditional second moment, meaning that its variance can be used to describe the magnitude of $\mu_{yct}(\bX, \bd{U})$. Since $\Var\{Y(t)\mid S(t),\bX\}=\sigmay^2+\Var\{\mu_{yct}\mid S(t),\bX\}$, and $\Var\{Y(t)\mid S(t),\bX\}$ is identifiable under standard assumptions, the residual variance $\sigmay^2$ can determine the magnitude of $\Var\{\mu_{yct}\mid S(t),\bX\}$. We can use the residual variance to control the contribution of the unobserved intermediate, with higher $\sigmay^2$ indicating less variability of $\mu_{yct}$ and lower $\sigmay^2$ indicating greater variability of $\mu_{yct}$. It is not clear what a reasonable bound for $\sigma_y^2$ is, however, we can use an extension of Assumption \ref{assump: dominant ob effect} to the nonparametric setting to guide this choice. Specifically, we can let $\eta_{A \sim B \mid  C}^2$ represent nonparametric partial $R^2$ values, which are defined as
$$\eta_{A \sim B \mid  C}^2 = \dfrac{\Var\{E(A\mid B, C)\}-\Var\{E(A\mid C)\}}{\Var(A)-\Var\{E(A\mid C)\}}.$$
We then make the following assumption on the contribution of the unobserved intermediate.
\begin{assumption}[\label{assump: generalized dominant ob effect}Generalized Dominant Observed Effect] 
Assume 
\begin{equation*}
        \eta^2_{Y\sim S(1-T)\mid S(T),T,\bX}\leq \eta^2_{Y\sim S(T)\mid T,\bX}.
\end{equation*}
\end{assumption}
Similar to before, this assumption places a restriction on the impact of the unobserved intermediate, which is analogous to restricting the magnitude of $\mu_{yct}$. This assumption implies the following inequality
\begin{equation*}
        \dfrac{(\Var \{Y(t)\}-\Var[E\{Y(t)\mid S(t),\bX\}])^2}{\Var\{Y(t)\}-\Var[E\{Y(t)\mid \bX\}]}\leq \sigma_y^2.
        \label{eq: EquivalentExpression}
\end{equation*}
This shows that we can control the variability of $\mu_{yct}$ through a lower bound on $\sigma_y^2$, similarly to what we did in the case of parametric, linear models. This is crucially important, because implementing a constraint on the variability of $\mu_{yct}$ is difficult to do in practice with complex models, but a constraint on $\sigma_y^2$ can easily be implemented within the Bayesian framework through a truncated prior distribution. Unlike in the linear, parametric outcome model, it is not straightforward to derive the implications of this assumption on the resulting size of the partial identification region for causal parameters of interest, though it is expected that the reduction will be large in many instances.

\section{Discussion} \label{sec:discussions}
This paper formalizes inference on principal causal effects when principal ignorability is violated by deriving partial identification results for outcome models when this assumption does not hold, and investigating the implications of this assumption on identification of the crucial association parameter $\rho$. We focused our results on settings with continuous intermediates as this is the scenario that relies most heavily on principal ignorability, though we discuss results for binary intermediates in Appendix C. Our results show that there is an inherent trade-off between the strength of assumptions made and the size of partial identification regions. If principal ignorability, or other assumptions, are not made, then all parameters are at best partially identified, and the widths of the partial identification regions for parameters of interest can be exceedingly wide. On the other hand, if strong assumptions are made, then all unknown parameters, including the typically unidentified $\rho$ parameter, can be consistently estimated. We have proposed alternative assumptions that should hold in most applications, which can be applied to sharpen inference and reduce the widths of partial identification regions under mild assumptions. Lastly, many of our results held for linear, parametric outcome models, but we showed that similar ideas can be applied for fully nonparametric outcome models, which are popular in causal inference.

There are a number of areas for future research that could build on this work. One area of interest may be the implementation of Model (\ref{model:BART}), along with an empirical examination of the extended sign and magnitude issues introduced by violations of principal ignorability. It would also be of interest to study the identification of $\rho$ in this scenario and whether this is possible under certain constraints on the outcome model functions. Additionally, it would be of interest to see if other, plausible assumptions could be used in conjunction with those seen in this manuscript to sharpen inference on partial identification regions. For example, we have not made the monotonicity assumption throughout because it is not sufficient for identification with continuous intermediates, though it is plausible in many applications. It may provide significant information on the distribution of $\boldsymbol{U}$ that could lead to smaller partial identification regions, or more efficient estimation of $\rho$ when it is identified.


\begin{appendix}
\section{Proofs}
\label{appendix:Proof}
\subsection{Proofs Relating to Model 
(\ref{model:ContinuousWithCovariates})}

Consider the following reparameterization for the outcome model in (\ref{model:SimpleBayesianModel}):
\begin{equation*}
    \begin{split} E[Y_{i}(t) \mid \bd{U}_i,\bX_i,T_i=t]&=\bbeta_{t}^T(\bd{U}_i-(\balpha^T\bX_i,\balpha^T\bX_i)^T)+\lambda_t+\bgamma^T_t\bX_i.
    \end{split}
    \label{model:ReparameterizedSimpleModel}
\end{equation*}
where $\bgamma_t = \bgamma + (\beta_{t0}+\beta_{t1})\balpha$. Given $T_i=t$, consider following change of variable:
\begin{equation}
    \label{eq: removing covariates}
    \begin{split}
        Y_i'(t)&=Y_i(t)-\bgamma_t^T\bX_i,\\
        S'_i(t)&=S_i(t)-\balpha^T\bX_i, \qquad t=0,1.
    \end{split}
\end{equation}

An important benefit of (\ref{eq: removing covariates}) is that $(\bgamma_0,\bgamma_1,\balpha)$ are identifiable in model (\ref{model:ContinuousWithCovariates}). A short explanation is as follows:
Apparently, $\balpha$ is identifiable. Since 
\begin{equation*}
\begin{split}
    E[Y_i(t)|S'_i(t),\bX_i,T_i=t]&=\lambda_t+\beta_{t,1-t}E[S'_i(1-t)|S'(t),T_i=t]+\beta_{tt}S'(t)+\bgamma_t^T\bX_i,
\end{split}
\end{equation*}
$\bgamma_t$ is also identifiable for $t=0,1$.

Then,
\begin{equation*}
    \begin{split}
        E[Y'_{i}\mid \bd{U}'_i,\bX_i,T_i=t]&=\bbeta_{t}^T\bd{U}'_i+\lambda_t,\\
        E[\bd{U}'_i\mid \bX_i]&= (\phi_0,\phi_1)^T.
    \end{split}
\end{equation*}

For proofs of Theorems \ref{thm:IdentifyRho} and \ref{thm:ApproximateVariance}, since
$(\bbeta_0,\bbeta_1,\balpha,\bgamma)$ all are assumed known, then we can apply (\ref{eq: removing covariates}). As for proofs of Propositions \ref{prop: PIR for byab and byba without constraints}, \ref{prop: PIR for byab and byba assuming same signs} and \ref{prop: PIR for byab and byba assuming dominant ob effect 1}, since $(\bgamma_0,\bgamma_1,\balpha)$ are identifiable, then we can also apply (\ref{eq: removing covariates}). 

\subsubsection{Proof of Theorem \ref{thm:IdentifyRho}}
 \label{appendix:ProofOfThm1}
\begin{proof}
    We apply (\ref{eq: removing covariates}) and let $Y'_i=Y'_i(T_i)$ and $S'_i=S'_i(T_i)$. Throughout the proof, we always condition on $\bX$ and $T$, unless specified.

    The conditional marginal pdf of $(Y'_i, S'_i)$ is 
    \begin{equation}
    (Y'_i,S'_i)\mid \boldsymbol{X}_i,T_i=t\sim N\left((\mu'_{yt},\phi_t),\begin{pmatrix}
    \zeta_t& \psi_t\sigma_{st} \\
    \psi_t\sigma_{st} &  \sigma^2_{st},
    \end{pmatrix}\right),
    \label{eq: marginal distribution after reparameterization}
    \end{equation}
    where $\mu'_{yt} = \bbeta_{t}^T\bd{\phi}_i+\lambda_t$, $\zeta_{t}=\sigma_y^2+\bbeta_t^T\bSigma_{s}\bbeta_t$, and $\psi_t=\sigma_{st}\beta_{tt}+\rho\sigma_{s,1-t}\beta_{t,1-t}$. 
    
   The conditional marginal pdf of $(Y'_1, \dots, Y'_n, S'_1, \dots, S'_n)$ given $\bd{X}_i=\bx_i$ and $T_i=t_i$ is
    \begin{equation*}
        p(\cdot)= \prod_{i=1}^n\bigg((2\pi\sigma_{s,t_i})^{-1}(\zeta_{t_i}-\psi_{t_i}^2)^{-1/2}\exp\bigg(-\frac{1}{2}\bigg(\frac{(y_i''-\psi_{t_i}s_i''/\sigma_{s,t_i})^2}{\zeta_{t_i}-\psi_{t_i}^2}+\frac{s_i''^2}{\sigma_{s,t_i}^2}\bigg)\bigg)\bigg),
    \end{equation*}
    where $y''_i=y'_i-\lambda_{t_i}-\bmu_{s}^T\bbeta_{t_i}$, $s''_i=s'_i-\phi_{t_i}$, and $\zeta_{t_i}-\psi_{t_i}^2=\sigma_y^2+(1-\rho^2)\sigma_{s,1-t_i}^2\beta_{t_i,1-t_i}^2$.

    Then, the unnormalized log posterior for $\rho$ is 
    \begin{equation*}
            q_n(\rho)=-\frac{1}{2}\suml{i=1}^n\left(\log(\zeta_{t_i}-\psi_{t_i}^2)+\frac{(y_i''-\psi_{t_i}s_i''/\sigma_{s,t_i})^2}{\zeta_{t_i}-\psi_{t_i}^2}+s_i''^2/\sigma_{s,t_i}^2\right)+\log(f(\rho)),
    \end{equation*}
    where $f(\rho)$ is the prior for $\rho$.
    
    Let 
    \begin{equation*}
        l_i(\rho\mid t)=-\frac{1}{2}\left(\log(\zeta_{t}-\psi_{t}^2)+\frac{(y_i''-\psi_{t}s_i''/\sigma_{st})^2}{(\zeta_{t}-\psi_{t}^2)}+ s_i''^2/\sigma_{st}^2\right),\qquad t=0,1.
    \end{equation*}

    Then, $q_n(\rho) = \suml{i=1}^n(t_il_i(\rho\mid 1)+(1-t_i)l_i(\rho\mid 0))+\log(f(\rho))/n$. Let $n_0=\suml{i=1}^n(1-t_i)$ and $n_1=\suml{i=1}^nt_i$. By the positivity assumption, $n_0\to\infty$ and $n_1\to\infty$, as $n\to\infty$. Within the control group ($T_i=0$) or the treatment group ($T_i=1$), $(Y'_i,S'_i)$ are i.i.d. 

    Let $\hat{\rho}_0=\argmax_\rho(\suml{i=1}^n(1-t_i)l_i(\rho\mid 0)/n_0)$ and $\hat{\rho}_1=\argmax_\rho(\suml{i=1}^nt_il_i(\rho\mid 1)/n_1)$. Next, we show that $\hat{\rho}_0\to\rhot$ and $\hat{\rho}_1\to\rhot$ in probability as $n\to\infty$. By Proposition 7.1 in \cite{hayashi_2011}, we need to show that they satisfy the identification condition and the uniform convergence condition within both groups.

    For the control group,
    \begin{itemize}
        \item Identification condition: 
        
        The conditional mean of $l_i(\rho\mid 0)$ is
        \begin{equation*}
        \begin{split}
            E(l_i(\rho\mid 0))&=-\frac{1}{2}\left(\log(\zeta_0-\psi_0^2)+\frac{(1-p)(\psi_0^2-2\psi_0\psit_0+\zetat_0)}{\zeta_0-\psi_0^2}\right),
        \end{split}
        \end{equation*}

        Taking derivative w.r.t $\rho$, we have 
        \begin{equation*}
        \begin{split}
            \frac{\partial E(l_i(\rho\mid 0))}{\partial \rho}&=\frac{2\sigmasb^2\byab^2}{\zeta_{0}-\psi_{0}^2}\left(\frac{-\sigmasb^2\byab^2\rho^2}{\zeta_{0}-\psi_{0}^2}-1\right)(\rho-\rho^*).
        \end{split}
        \end{equation*}

        As long as $\byab\neq0$, we have $2\sigmasb^2\byab^2/(\zeta_{0}-\psi_{0}^2)(-\sigmasb^2\byab^2\rho^2/(\zeta_{0}-\psi_{0}^2)-1)<0$, where $\zeta_{0}-\psi_{0}^2 = \sigmay^2+(1-\rho^2)\byab^2\sib^2>0$. Thus, $E(l_i(\rho\mid0))$ is strictly increasing on $[-1,\rhot)$ and strictly decreasing on $(\rhot,1]$, indicating that $E(l_i(\rho\mid 0))$ is uniquely maximized on $[-1,1]$ at $\rhot$.

    \item Uniform convergence: 
    
    Note that $\Var\left(\suml{i=1}^n(1-t_i)l_i(\rho\mid 0)/n_0\right)=\Var(l_i(\rho\mid 0))/n_0=4/n_0$. By Chebyshev's inequality, uniform convergence holds.
        
    \end{itemize}

    For the treatment group, similarly, we can show that these also hold.

    Since $q_n({\rho})= \suml{i=1}^n\left(t_il_i(\rho\mid 1)/n_0(n_0/n)+(1-t_i)l_i(\rho\mid 0)/n_1(n_1/n)\right)+\log(f(\rho))/n$, the above provides intuition that $\hat{\rho}_B$ (posterior mode) converges to $\rhot$. To complete our proof rigorously, we will show that for any $\epsilon>0$, $P(|\hat{\rho}_B-\rhot| >\epsilon)\to0$ as $n\to\infty$. If $f(\rho)$ is continuous or bounded, then $\log(f(\rho))/n\to 0$ uniformly as $n\to\infty$. Since $\suml{t_i=t}l_i(\rho\mid t)/n_t$ uniformly converges to $E(l_i(\rho\mid t))$, we have $| q_n({\rho})-E(l_i(\rho\mid 0))(n_0/n)-E(l_i(\rho\mid 1))(n_1/n)| \to 0$ uniformly in probability, and $P(q_n(\rho)\geq q_n(\rhot):|\rho-\rhot| >\epsilon)\to 0$, that is, $P(|\hat{\rho}_B-\rhot| >\epsilon)\to0$.

\end{proof}

\subsubsection{Proof of Theorem \ref{thm:ApproximateVariance}}
\begin{proof}
    We inherit the notations from the proof of Theorem \ref{thm:IdentifyRho}. Throughout the proof, we always condition on $\bX$ and $T$, unless specified. In here we use the Laplace approximation to obtain the asymptotic approximation of posterior variance. Note that the Bernstein–von Mises Theorem provide the justifications for the Laplace approximation.

    \begin{equation*}
    \begin{split}
        E(\rho\mid \Dc)&\approx\frac{\int_{\rho^*-\delta_1}^{\rho^*+\delta_2} \rho \exp(q_n(\hat{\rho}_B)+\frac{1}{2}q_n''(\hat{\rho}_B)(\rho-\hat{\rho}_B)^2+o((\rho-\hat{\rho}_B)^2)d\rho}{\int_{\rho^*-\delta_1}^{\rho^*+\delta_2}  \exp(q_n(\hat{\rho}_B)+\frac{1}{2}q_n''(\hat{\rho}_B)(\rho-\hat{\rho}_B)^2+o((\rho-\hat{\rho}_B)^2)d\rho}\\
        &=\hat{\rho}_B,
    \end{split}
    \end{equation*}
    where $q_n''=\dfrac{\partial^2 q_n(\rho)}{\partial \rho^2}$.
    
    Similarly,
    \begin{equation*}
    \begin{split}
        E(\rho^2\mid \Dc)&\approx\frac{\int_{\rho^*-\delta_1}^{\rho^*+\delta_2} \rho^2 \exp(q_n(\hat{\rho}_B)+\frac{1}{2}q_n''(\hat{\rho}_B)(\rho-\hat{\rho}_B)^2+o((\rho-\hat{\rho}_B)^2)d\rho}{\int_{\rho^*-\delta_1}^{\rho^*+\delta_2}  \exp(q_n(\hat{\rho}_B)+\frac{1}{2}q_n''(\hat{\rho}_B)(\rho-\hat{\rho}_B)^2+o((\rho-\hat{\rho}_B)^2)d\rho}\\
        &=\hat{\rho}_B^2 - \frac{1}{q_n''(\hat{\rho}_B)}.
    \end{split}
    \end{equation*}

    Therefore, an approximation of the posterior variance is $-\frac{1}{q_n''(\hat{\rho}_B)}$. By the Strong Law of Large Numbers (SLLN), we have 
    \begin{equation*}
        \suml{t_i=t}l''_i(\rho\mid t)n_t^{-1}\to E(l''_i(\rho\mid t)), \text{ almost surely,}
    \end{equation*}
    where $l''_i(\rho\mid t)=\frac{\partial^2 l_i(\rho|t)}{\partial \rho^2}$.
    
    Since $q''_n(\rho)/n = \suml{t_i=0}(l''_i(\rho\mid 0)/n_0) (n_0/n)+\suml{t_i=1}(l''_i(\rho\mid 1)/n_1) (n_1/n)$, then an asymptomatic approximation of $q_n''(\rho)/n$ can be 
    \begin{equation*}
        \begin{split}
            q''_n(\rho)/n \approx E(l''_i(\rho\mid 0)) (1-\Bar{T})+E(l''_i(\rho\mid 1)) \Bar{T},
        \end{split}
    \end{equation*}
    where $\Bar{T}=\suml{i=1}^nT_i/n$.

    Note that there is only a constant difference (w.r.t $\rho$) between $l''_i(\rho\mid t)$ and the log-likelihood of $(Y'_t,S'_t)$. Applying the second Bartlett's Identity to it, we have
    \begin{equation}
    \begin{split}
        Var(\rho\mid \Dc)&\sim n^{-1}\left(\Bar{T}E\bigg( \Big(\frac{\partial l_i(\rho|1)}{\partial \rho} \Big)^2\bigg)+(1-\Bar{T})E\bigg( \Big(\frac{\partial l_i(\rho|0)}{\partial \rho} \Big)^2\bigg)\right)^{-1}\bigg| _{\rho=\hat{\rho}_B}\\
        &\sim n^{-1}\bigg(\Bar{T}\byba^2\sia^2\bigg(\frac{2\rho^2\byba^2\sia^2}{(\zeta_1-\psi_1^2)^2}+\frac{1}{\zeta_1-\psi_1^2}\bigg)+\\
        &\qquad (1-\Bar{T})\byab^2\sib^2\bigg(\frac{2\rho^2\byab^2\sib^2}{(\zeta_0-\psi_0^2)^2}+\frac{1}{\zeta_0-\psi_0^2}\bigg)\bigg)^{-1}\bigg| _{\rho=\hat{\rho}_B}.
    \end{split}
    \label{eq: asymptotic posterior variance}
    \end{equation}

    Since the right-hand side of (\ref{eq: asymptotic posterior variance}) is a continuous function of $\rho$, and $\hat{\rho}_B$ is consistent for $\rhot$, by the continuous mapping theorem, we can replace $\hat{\rho}_B$ with $\rhot$ asymptotically.
\end{proof}

For Proposition \ref{prop: PIR for byab and byba without constraints}, \ref{prop: PIR for byab and byba assuming same signs} and \ref{prop: PIR for byab and byba assuming dominant ob effect 1}, we present the proof considering covariates.

\subsubsection{Proof of Proposition \ref{prop: PIR for byab and byba without constraints}}
\begin{proof}
   We apply (\ref{eq: removing covariates}) and let $Y'_i=Y'_i(T_i)$ and $S'_i=S'_i(T_i)$. For the marginal parameters in (\ref{eq: marginal distribution after reparameterization}), we have
    \begin{equation}
    \begin{cases}
      &\mu_{y0}'= \lambda_0+\beta_{00}\phi_0+\beta_{01}\phi_1\\
      &\mu_{y1}'= \lambda_1+\beta_{10}\phi_0+\beta_{11}\phi_1\\
      &\zeta_0= \sigma_y^2+\bbeta_0^T\bSigma_{s}\bbeta_0\\
      &\zeta_1= \sigma_y^2+\bbeta_1^T\bSigma_{s}\bbeta_1\\
      &\psi_0=\sigma_{s0}\beta_{00}+\rho\sigma_{s1}\beta_{01}\\
      &\psi_1=\sigma_{s1}\beta_{11}+\rho\sigma_{s0}\beta_{10}
    \end{cases} 
    \label{eq: constraints for joint prameters}
    \end{equation}
    The general solutions of (\ref{eq: constraints for joint prameters}) is
    \begin{equation}
    \begin{cases}
      &\beta_{t,1-t}^2=(\zeta_t-\psi_t^2-\sigma_y^2)/(1-\rho^2)/\sigma^2_{s,1-t}\\
      &\beta_{tt}=(\psi_t-\rho\sigma_{s,1-t}\beta_{t,1-t})/\sigma_{st}\\
      &\lambda_t=\mu'_{yt}-(\beta_{t0}\phi_0+\beta_{t1}\phi_{1}).
    \end{cases} 
    \label{eq: solutions to constraints equations}
    \end{equation}
    In (\ref{eq: solutions to constraints equations}), except for signs, there is only one unknown parameter $\sigma_y^2$. Since $\beta_{t,1-t}^2\geq0$, (\ref{eq: solutions to constraints equations}) also implies that $0\leq\sigmay^2\leq \min\limits_{t=0,1} (\zeta_t-\psi_t^2)$. Plugging this into the first equation of (\ref{eq: solutions to constraints equations}), we have 
    \begin{equation*}
        \frac{\zeta_t-\psi_t^2-\min\limits_{t=0,1} (\zeta_t-\psi_t^2)}{(1-\rho^2)\sigma^2_{s,1-t}}\leq\beta_{t,1-t}^2\leq\frac{\zeta_t-\psi_t^2}{(1-\rho^2)\sigma^2_{s,1-t}}.
    \end{equation*}
    
    Since $\sib^2\byab^2\leq\sia^2\byba^2$ (which is equivalent to $\Var(Y(1)\mid S(1),\bd{X})\geq\Var(Y(0)\mid S(0),\bd{X})$), then it can be simplified as
    \begin{equation*}
        \begin{split}
            0\leq&\beta_{01}^2\leq\frac{\zeta_0-\psi_0^2}{(1-\rho^2)\sigma^2_{s1}},\\
            \frac{\zeta_1-\psi_1^2-(\zeta_0-\psi_0^2)}{(1-\rho^2)\sigma^2_{s0}}\leq&\beta_{10}^2\leq\frac{\zeta_1-\psi_1^2}{(1-\rho^2)\sigma^2_{s0}}.
        \end{split}
    \end{equation*}

    Note that (\ref{eq: solutions to constraints equations}) also requires that
    \begin{equation*}
        \sia^2\byba^2-\sib^2\byab^2=\frac{\zeta_1-\psi_1^2-(\zeta_0-\psi_0^2)}{1-\rho^2}.
    \end{equation*}

    We plug in $\Var(Y(t)\mid S(t),\bd{X})=\Var(Y(t)\mid S(t),\bd{X})=\zeta_t-\psi_t^2$ and $\Var(S(t)\mid\bd{X})=\sigma_{st}^2$ to obtain the desired forms. 
\end{proof}

\subsubsection{Derivation for (\ref{eq:PCE in terms of beta})}
\begin{equation*}
    \begin{split}
        PCE(\bd{u})&=E(PCE(\bd{u},\bd{x})\mid\bd{U}=\bd{u})\\
        &=E((\byba-\byaa)s_0+(\bybb-\byab)s_1+\lambda_1-\lambda_0\mid\bd{U}=\bd{u})\\
        &=(\byba-\byaa)s_0+(\bybb-\byab)s_1+\lambda_1-\lambda_0
    \end{split}
\end{equation*}

Using (\ref{eq: solutions to constraints equations}), we first replace $\lambda_t$ and then replace $\beta_{tt}$. After simplification, we have

\begin{equation*}
    \begin{split}
        PCE(\bd{u})&=(\byba-\frac{\psi_0}{\sia}+\frac{\sib}{\sia}\rho\byab)(s_0-\phi_0)+\\
        &\qquad (\frac{\psi_1}{\sib}-\frac{\sia}{\sib}\rho\byba-\byab)(s_1-\phi_1)+(\mu'_{y1}-\mu'_{y0}),
    \end{split}
\end{equation*}

\subsubsection{Proof of Proposition \ref{prop: PIR for byab and byba assuming same signs}}
\begin{proof}
    By directly applying Assumption \ref{assump: same sign} to Proposition \ref{prop: PIR for byab and byba without constraints}, we can obtain Proposition \ref{prop: PIR for byab and byba assuming same signs}.
\end{proof}

\subsubsection{Proof of Proposition \ref{prop: PIR for byab and byba assuming dominant ob effect 1}}
\begin{proof}
    Assumption \ref{assump: dominant ob effect} is equivalent to
    \begin{equation*}
    \begin{split}
        \frac{(\Cov(Y(t),S(1-t)\mid S(t),\boldsymbol{X}))^2}{\Var(S(1-t)\mid S(t),\boldsymbol{X})\Var(Y\mid S(t),\boldsymbol{X})}&\leq \frac{(\Cov(Y(t),S(t)\mid \boldsymbol{X}))^2}{\Var(Y\mid \boldsymbol{X})\Var(S(t)\mid \boldsymbol{X})}\\
        \frac{\beta_{t,1-t}^2\sigma^4_{s,1-t}(1-\rho^2)^2}{\sigma^2_{s,1-t}(1-\rho^2)(\sigma_y^2+\beta_{t,1-t}^2\sigma^2_{s,1-t}(1-\rho^2))}&\leq\frac{(\beta_{tt}\sigma_{st}+\rho\beta_{t,1-t}\sigma_{s,1-t})^2}{\sigma_{st}^2(\sigma_y^2+\bbeta_t^T\bSigmasi\bbeta_{t})}\\
    \end{split}
    \end{equation*}

    Using (\ref{eq: solutions to constraints equations}), we have 
    \begin{equation*}
    \begin{split}
        \frac{(\zeta_t-\psi_t^2)^2}{\zeta_t}&\leq\sigma_y^2.
    \end{split}
    \end{equation*}

    Combing it with the proof for Proposition \ref{prop: PIR for byab and byba without constraints}, the corresponding partial identification region can be obtained.
\end{proof}

Note that the above proof also shows that Assumption \ref{assump: dominant ob effect} is equivalent to 
    \begin{equation*}
        \dfrac{(\Var(Y(t)\mid S(t), \boldsymbol{X}))^2}{\Var(Y(t)\mid\boldsymbol{X})}\leq\Var(Y(t)\mid S(t),S(1-t),\boldsymbol{X}).
    \end{equation*}

\subsection{Proofs Relating to Model (\ref{model:BART})}
\subsubsection{Proof of Theorem \ref{thm:UnidentifiedInvariant}}
\begin{proof}

    Note that $Y'(t)$ and $S'(1-t)$ are also used in this proof, but they have different definitions than before.. Let $f_t(S(t),S(1-t),\boldsymbol{X})$ denote $\mu_{yt}(\bX,S(t))+\mu_{yct}(\bX,S(t),S(1-t))$. Then 
    \begin{equation*}
        f_t\big(S(t),M^{-1}(S(1-t)),\boldsymbol{X}\big)=\mu_{yt}\big(\bX,S(t)\big)+\mu_{yct}\big(\bX,S(t),M^{-1}(S(1-t))\big),
    \end{equation*}
    and 
    \begin{equation*}
        Y'(t)=f_t(S(t),M^{-1}(S(1-t)),\boldsymbol{X})+\epsilon.
    \end{equation*}
    
    Although $M$ could relate to $S(t)$ and $\boldsymbol{X}$, we do not include them in the $M$ for simplicity of notation. Let $S'(1-t)$ denote $M(S(1-t))$ and $p_{\epsilon}$ denote the pdf of $\epsilon$. Throughout, for simplicity of notation we let $p\big(Y'(t)=y)$ denote the pdf of $Y'(t)$ evaluated at $y$, with similar notation for $S(t)$.  The the conditional distribution of $Y'(t)$ given $S(t)$ and $\boldsymbol{X}$ is 
    \begin{equation*}
        \begin{split}
            p\big(Y'(t)=y\mid \boldsymbol{X},S(t)\big)&=\int p\big(Y'(t)=y\mid S(1-t)=s_{1-t},S(t),\boldsymbol{X}\big)
        dP_{S(1-t)\mid S(t),\boldsymbol{X}}\big(s_{1-t}\big) \\
        &=\int p_{\epsilon}\big(y-f_t(S(t),M^{-1}(S(1-t)),\boldsymbol{X})\big)dP_{S(1-t)\mid S(t),\boldsymbol{X}}\big(s_{1-t}\big)\\
        &=\int p_{\epsilon}\big(y-f_t(S(t),M^{-1}(s'_{1-t}),\boldsymbol{X})\big)dP_{S'(1-t)\mid S(t),\boldsymbol{X}}\big(s'_{1-t}\big)\\
        &=\int p_{\epsilon}\big(y-f_t(S(t),s_{1-t},\boldsymbol{X})\big)dP_{S(1-t)\mid S(t),\boldsymbol{X}}\big(s_{1-t}\big)\\
        &=p\big(Y(t)=y\mid \boldsymbol{X},S(t)\big).
        \end{split}
    \end{equation*}

    The third equality holds since $S'(1-t)\mid S(t),\boldsymbol{X}$ and $S(1-t)\mid S(t),\boldsymbol{X}$ have the same distribution and we also change the notation from $s_{1-t}$ to $s'_{1-t}$. The fourth equality holds since we apply the change of variable $S(1-t)=M^{-1}(S'_{1-t})$.

    The the conditional distribution of $(Y'(t),S(t))$ given $\boldsymbol{X}$ is 

    \begin{equation*}
        \begin{split}
        p\big(Y'(t)=y,S(t)=s_t\mid \boldsymbol{X}\big)&=p\big(Y'(t)=y\mid \boldsymbol{X},S(t)=s_t\big)p\big(S(t)=s_t\mid \bd{X}\big)\\
        &=p\big(Y(t)=y\mid \boldsymbol{X},S(t)=s_t\big)p\big(S(t)=s_t\mid \bd{X}\big)\\
        &=p\big(Y(t)=y,S(t)=s_t\mid \boldsymbol{X}\big).
        \end{split}
    \end{equation*}

    Therefore, $(Y'(t),S(t))\mid \bd{X}$ and $(Y(t),S(t))\mid \bd{X}$ follow the same distribution.

\end{proof}

\subsubsection{Proof of an equivalent expression for Assumption \ref{assump: generalized dominant ob effect}}

\begin{proof}
    Using the definition of Pearson's correlation ratio, Assumption \ref{assump: generalized dominant ob effect} is equivalent to:
    \begin{equation*}
        \begin{split}
            &\dfrac{\Var(E[Y(t)\mid \bd{U},\bX])-\Var(E[Y(t)\mid S(t),\bX])}{\Var(Y(t))-\Var(E[Y(t)\mid S(t),\bX])} \leq\\
            &\qquad \dfrac{\Var(E[Y(t)\mid S(t),\bX])-\Var(E[Y(t)\mid \bX])}{\Var(Y(t))-\Var(E[Y(t)\mid \bX])}.
        \end{split}
    \end{equation*}

    Since $\epsilon \ind E[Y(t)\mid \bd{U},\bX]$, then 
    \begin{equation*}
        \Var(Y(t)) = \Var(E[Y(t)\mid \bd{U},\bX]) + \sigmay^2.
    \end{equation*}

    Combining them, we have 
    \begin{equation*}
        \begin{split}
            &\dfrac{\Var(Y(t)) - \sigmay^2-\Var(E[Y(t)\mid S(t),\bX])}{\Var(Y(t))-\Var(E[Y(t)\mid S(t),\bX])} \leq\\
            &\qquad \dfrac{\Var(E[Y(t)\mid S(t),\bX])-\Var(E[Y(t)\mid \bX])}{\Var(Y(t))-\Var(E[Y(t)\mid \bX])}.
        \end{split}
    \end{equation*}

    We simplify the above, then we have
    \begin{equation*}
        \dfrac{(\Var(Y(t))-\Var(E[Y(t)\mid S(t),\bX]))^2}{\Var(Y(t))-\Var(E[Y(t)\mid \bX])}\leq \sigmay^2
    \end{equation*}
\end{proof}

\section{Details of MCMC Sampling}
 \label{appendix:DetailsOfMCMC}
    In this section, we detail the general Gibbs updating steps for Model (\ref{model:ContinuousWithCovariates}). Note that Model (\ref{model:ContinuousWithCovariates}) includes covariates. Throughout this section, let $\bd{\theta}_y=(\bbeta_0,\bbeta_1,\lambda_0,\lambda_1,\bgamma)$ and $\bd{\theta}_s=(\bphi,\balpha)$, both of which are column vectors. The full joint pdf of is:
\begin{equation}
    \begin{split}
        &\dfrac{1}{(2\pi)^{\frac{3n}{2}}\sigma_y^{n}\sia^n\sib^n(1-\rho^2)^{\frac{n}{2}}}\exp\bigg(-\frac{1}{2}\suml{i=1}^n\bigg((\btheta_y^T\bd{d}_i-y_i)^2/\sigma_y^2+\\
        &\qquad(\bd{E}_i^T\btheta_s-\bd{u}_i)^T\bSigmasi^{-1}(\bd{E}_i^T\btheta_s-\bd{u}_i)\bigg)\bigg),
    \end{split}
\end{equation}
 where 
    \begin{equation*}
        \bd{d}_i=\begin{pmatrix}
            \bd{u}_i(1-t_i)\\
            \bd{u}_it_i\\
            1-t_i\\
            t_i\\
            \bx_i
        \end{pmatrix}\ \text{and}\ \bd{E}_i=(\bd{e}_{i0},\bd{e}_{i1})=\begin{pmatrix}
            1&0\\
            0&1\\
            \bx_i&\bx_i\\
        \end{pmatrix}.
    \end{equation*}

We consider conjugate priors for all parameters except $\rho$ as follows: 

\begin{equation*}
    \begin{split}
        &\bbeta_t\sim N(\bmu_{\bbeta_t},\Sigma_{\bbeta_t}),\ 
        \lambda_t\sim N(\mu_{\lambda_t},\sigma_{\lambda_t}^2),\ 
        \phi_t\sim N(\mu_{\phi_t},\sigma_{\phi_t}^2),\ \bgamma\sim N(\bmu_{\bgamma},\Sigma_{\bgamma}),\\
        &\balpha\sim N(\bmu_{\balpha},\Sigma_{\balpha}),\ \sigma_y^2\sim IG(\eta_y,\nu_y),\ 
        \sigma_{s0}^2\sim IG(\eta_{s0},\nu_{s0}),\ 
        \sigma_{s1}^2\sim IG(\eta_{s1},\nu_{s1}),
    \end{split}
\end{equation*}
and a flat prior for $\rho$.

Then, the updating steps for Gibbs sampling are as follows
\begin{enumerate}[label=\alph*.]
    \item We can update the unobserved values of the intermediate $S(1-t_i)$ from
    \begin{equation*}
        S(1-t_i)\mid \cdot\sim N(\mu_{si,mis},\sigma_{si,mis}^2),
    \end{equation*}
    where 
    \begin{equation*}
        \sigma_{si,mis}^2=\left(\frac{\beta_{t_i,1-t_i}^2}{\sigma_y^2}+\frac{1}{\sigma_{s,1-t_i}^2(1-\rho^2)}\right)^{-1},
    \end{equation*}
    and 
    \begin{equation*}
        \begin{split}
            \mu_{si,mis}&=\left(\frac{\beta_{t_i,1-t_i}^2}{\sigma_y^2}+\frac{1}{\sigma_{s,1-t_i}^2(1-\rho^2)}\right)^{-1}\bigg(-(\beta_{t_i,t_i}s_{i}+\lambda_{t_i}+\bgamma^T\bx_i-y_i)\beta_{t_i,1-t_i}/\sigma_y^2+\\
            &\qquad (\phi_{1-t_i}+\balpha^T\bx_i)/\sigma_{s,1-t_i}^2/(1-\rho^2)+\rho(s_i-\phi_{t_i}-\balpha^T\bx_i)/\big(\sigmasa\sigmasb(1-\rho^2)\big)\bigg).
        \end{split}
    \end{equation*}

    \item
    We can update $\bd{\theta}_y$ from the following conditional distribution
    {\small\begin{equation*}
    \bd{\theta}_y\mid \cdot\sim N\left(\bigg(\suml{i=1}^n\bd{d}_{i}\bd{d}_i^T/\sigma_y^2+\Sigma_{\btheta_y}^{-1}\bigg)^{-1}\bigg(\suml{i=1}^ny_i\bd{d}_i/\sigma_y^2+\Sigma_{\btheta_y}^{-1}\bmu_{\btheta_y}\bigg),\bigg(\suml{i=1}^n\bd{d}_{i}\bd{d}_i^T/\sigma_y^2+\Sigma_{\btheta_y}^{-1}\bigg)^{-1}\right),
    \end{equation*}}
    where 
    \begin{equation*}
        \bSigma_{\btheta_y}=\text{diag}(\bSigma_{\bbeta_0},\bSigma_{\bbeta_1},\sigma^2_{\lambda_0},\sigma^2_{\lambda_1},\bSigma_{\bgamma}).
    \end{equation*}
    \item We update $\bd{\theta}_s$ from
    {\small\begin{equation*}
    \bd{\theta}_s\mid \cdot\sim N\left(\bigg(\suml{i=1}^n\bd{E}_{i}\bSigma_{s}^{-1}\bd{E}_i^T+\Sigma_{\btheta_s}^{-1}\bigg)^{-1}\bigg(\suml{i=1}^n\bd{E}_i\bSigma_{s}^{-1}\bd{u}_i+\Sigma_{\btheta_s}^{-1}\bmu_{\btheta_s}\bigg),\bigg(\suml{i=1}^n\bd{E}_{i}\bSigma_{s}^{-1}\bd{E}_i^T+\Sigma_{\btheta_s}^{-1}\bigg)^{-1}\right),
    \end{equation*}}
    where $\bSigma_{\btheta_s}=\text{diag}(\sigma_{\phi_0}^2,\sigma_{\phi_1}^2,\bSigma_{\balpha})$.
    
    \item We can update $\sigmay$ from
    \begin{equation*}
        \sigma_y^2\mid \cdot\sim IG\left(\frac{n}{2}+\etay,\suml{i=1}^n(\btheta_y^T\bd{d}_i-y_i)^2)/2+\nuy\right).
    \end{equation*}

    \item If we consider $\sigmasa\neq\sigmasb$, then we need to apply Metropolis-Hasting algorithm to update $\sigmasa^2$ and $\sigmasb^2$ from 
    \begin{enumerate}
        \item
    \begin{equation*}
        \begin{split}
            p(\sia^2\mid \cdot)&\propto \sia^{-n-2-2\etasa}\exp\bigg(-\frac{1}{2(1-\rho^2)\sib^2}\suml{i=1}^n\bigg((\btheta_{s}^T\beia-s_{i0})^2\sib^2/\sia^2-\\
            &\qquad 2\rho\sib(\btheta_{s}^T\beia-s_{i0})(\btheta_{s}^T\beib-s_{i1})/\sia\bigg)-\frac{\nusa}{\sigmasa^2}\bigg).
        \end{split}
    \end{equation*}
    \item
    \begin{equation*}
        \begin{split}
            p(\sib^2\mid \cdot)&\propto \sib^{-n-2-2\etasb}\exp\bigg(-\frac{1}{2(1-\rho^2)\sia^2}\suml{i=1}^n\bigg((\btheta_{s}^T\beib-s_{i1})^2\sia^2/\sib^2-\\
            &\qquad 2\rho\sia(\btheta_{s}^T\beia-s_{i0})(\btheta_{s}^T\beib-s_{i1})/\sib\bigg)-{\nusb}/{\sigmasb^2}\bigg).
        \end{split}
    \end{equation*}
    \end{enumerate}

    If we assume $\sia=\sib=\sigma_s$, then given a conjugate inverse gamma prior, $\sigma_s^2\sim IG(\eta_s,\nu_s)$, we can update $\sigma_s^2$ using the following conditional distribution:
    \begin{equation*}
        \begin{split}
            \sigma_s^2\mid IG\bigg(n+\eta_s,\suml{i=1}^n(\bd{E}_i^T\btheta_s-\bd{u}_i)^T\bd{R}_s^{-1}(\bd{E}_i^T\btheta_s-\bd{u}_i)/2+\nu_s\bigg),
        \end{split}
    \end{equation*}
    where 
    \begin{equation*}
        \bd{R}_s=\begin{pmatrix}
            1&\ \rho\\
            \rho&\ 1
        \end{pmatrix}.
    \end{equation*}
    
    \item
    Instead of updating $\rho$ from posterior of full data including unobserved intermediates, we update $\rho$ using marginal posterior with observed intermediate variables only. Then, we have
        \begin{equation*}
        \begin{split}
         p(\rho\mid \cdot)
        &\propto \prod_{i=1}^n\bigg((\zeta_{t_i}-\psi_{t_i}^2)^{-1/2}\exp\bigg(-\frac{1}{2}\suml{i=1}^n\frac{(\psi_{t_i}s''_i/\sigma_{st}-y''_i)^2}{\zeta_{t_i}-\psi_{t_i}^2}\bigg)\bigg),
        \end{split}
        \end{equation*}
    where given $T_i=t$, $y''_i=y_i-\lambda_t-\bgamma^T\bx_i-\bmu_{si}^T\bbeta_t$ and $s''_i=s_i-\phi_t-\balpha^T\bx_i$. Since the conditional distribution of $\rho$ is intractable, we apply a Metropolis-Hastings algorithm to update $\rho$.
\end{enumerate}

If there are no covariates, one can simply remove $\bx_i$ and corresponding parameters $(\balpha,\bgamma)$. For truncated priors, one can modify the above Gibbs updating steps with truncated versions of the corresponding distributions.

\section{Results for Binary Intermediates}
\label{appendix:BinaryIntermediates}

\newcommand{\piaa}{p_{i00}}
\newcommand{\piab}{p_{i01}}
\newcommand{\piba}{p_{i10}}
\newcommand{\pibb}{p_{i11}}
\newcommand{\pa}{p_{i0}}
\newcommand{\pb}{p_{i1}}
\newcommand{\uia}{s_{i0}}
\newcommand{\uib}{s_{i1}}


Consider the following model:
\begin{equation}
    \begin{split}
        Y_{i}|\bd{U_i},\bX_i,T_i=t&\sim N(\mu_{yit},\sigma_y^2),\qquad t=0,1.\\
        p_{i,\uia,\uib}=P(\bd{U}_i=(\uia,\uib)|\bX_i)&=\piaa^{(1-\uia)(1-\uib)}\piab^{(1-\uia)\uib}\piba^{\uia(1-\uib)}\pibb^{\uia\uib},
        \label{model:Binary}
    \end{split}
\end{equation}
where $(\piaa,\piab,\piba,\pibb)$ are functions of $X_i$. Because the marginal probabilities of the two-way contingency table are identifiable, knowing any one of $(\piaa,\piab,\piba,\pibb)$ determines the others, and therefore we treat $\pibb$ as the only unknown parameter. We let $\pibb=f(\bd{x_i};\bd{\zeta},\rho)$, where $\bd{\zeta}$ denotes parameters that can be identified marginally, and $\rho$ denotes an (association) parameter that cannot be identified marginally. For example, we can let $\rho$ be the correlation parameter between  $S_{0}$ and $S_{1}$. 

\newcommand{\paa}{p_{00}}
\newcommand{\pab}{p_{01}}
\newcommand{\pba}{p_{10}}
\newcommand{\pbb}{p_{11}}
\newcommand{\psa}{p_{1\cdot}}
\newcommand{\psb}{p_{\cdot1}}

For simplicity of discussion, we consider the scenario where there are no covariates for the remainder of this section.  Since identification of $\pbb$ is equivalent to identification of $\rho$, we focus on identification of $\pbb$ (note that for simplicity we will use $\pbb$ and $\rho$ exchangeably in the rest of this section.). By knowing the marginal probabilities, we have
\begin{equation*}
    \begin{split}
        \pba&=\psa-\pbb,\\
        \pab&=\psb-\pbb,\\
        \paa&=1-\psa-\psb+\pbb,
    \end{split}
\end{equation*}
where $\psa=P(S(0)=1)$ and $\psb=P(S(1)=1)$.
Thus, assuming a flat prior, the unnormalized posterior of $\rho$ is 
\begin{equation*}
    \begin{split}
        P(\rho \vert \mathcal{D})&\propto\prod_{i=1}^n(\suml{s_{i,1-t_i}=0,1} \exp(-\frac{(y-\lambda_{t_i}-\bbeta_{t_i}^T\bd{u})^2}{2\sigmay^{2}})\paa^{(1-\uia)(1-\uib)}\pab^{(1-\uia)\uib}\pba^{\uia(1-\uib)}\\
        &\qquad \pbb^{\uia\uib}).
    \end{split}
\end{equation*}

\subsection{Identification of the Association Between Principal Strata}
\begin{proposition}
    Suppose that $\btheta=(\bbeta_0,\bbeta_1,\lambda_0,\lambda_1,\sigma_y^2,\psa,\psb)$ are known and principal ignorability fails: that is, at least one of $\beta_{01}$ or $\beta_{10}$ is nonzero. Under model (\ref{model:Binary}), the posterior mode of $\rho \mid \mathcal{D}$ is consistent.
\end{proposition}

\begin{proof}
    Let $f_t(y_i,s_i)$ denote the conditional marginal pdf of $(Y_i,S_i)|\rho,T_i=t$, for $t=0,1$. For simplicity, at times we simply refer to this as $f_t$. Let $h_n(\rho)$ denote the log normalized posterior, defined as $h_n(\rho)=\suml{i=1}^n\log
(f_{t_i}(y_i,s_i|\rho))$, and $q_n(\rho)=h_n(\rho)/n$. Let $q_0(\rho)=P(T_i=1)E(\log(f_1))+P(T_i=0)E(\log(f_0))$.

The conditional marginal pdf of $(Y_i,S_i)|\rho,T_i=t$, for $t=0$, is
\begin{equation}
    \begin{split}
        f_0(y,s|\rho) &= \dfrac{1}{\sqrt{2\pi}\sigmay} \exp(-\dfrac{(y-\lambda_0-\byaa s)^2}{2\sigmay^2})(\paa^{1-s}\pba^{s}+\\
        &\qquad\exp(-\dfrac{\byab^2-2\byab(y-\lambda_0-\byaa s)}{2\sigmay^2})\pab^{1-s}\pbb^{s}),
        \label{eq:BinaryMarginalPDF0}
    \end{split}
\end{equation}
and, for $t=1$, is 
\begin{equation*}
    \begin{split}
        f_1(y,s|\rho) &= \dfrac{1}{\sqrt{2\pi}\sigmay} \exp(-\dfrac{(y-\lambda_1-\bybb s)^2}{2\sigmay^2})(\paa^{1-s}\pab^{s}+\\
        &\qquad\exp(-\dfrac{\byba^2-2\byba(y-\lambda_1-\bybb s)}{2\sigmay^2})\pba^{1-s}\pbb^{s}),
    \end{split}
\end{equation*}
where $\rho=\pbb$. Similar to Appendix \ref{appendix:ProofOfThm1}, we only need to check identification for $q_0$ and uniform convergence for $q_n$.

\begin{itemize}
    \item Identification condition:  
    
    Since $\log(\cdot)$ is a strictly concave function, by Jensen's inequality we have
\begin{equation*}
    \begin{split}
        E_{\rho^\ast}[\log(\frac{f_0(Y,S|\rho)}{f_0(Y,S|\rho^\ast)})|T_i=0]&\leq \log(E_{\rho^\ast}[\frac{f_0(Y,S|\rho)}{f_0(Y,S|\rho^\ast)}|T_i=0])=0\\E_{\rho^\ast}[\log(f_0(Y,S|\rho))|T_i=0]&\leq E_{\rho^\ast}[\log(f_0(Y,S|\rho^\ast))|T_i=0],
    \end{split}
\end{equation*}
where the equality holds if and only if $\rho$ satisfies $P(f_0(Y,S|\rho)=f_0(Y,S|\rho^\ast))=1$ or $\byab=0$. When $\byab\neq0$, since $P(\exp(-\dfrac{\byab^2-2\byab(Y-\lambda_0-\byaa S)}{2\sigmay^2}) \neq 1)=1$, then $P(f_0(Y,S|\rho)=f_0(Y,S|\rho^\ast))=1$ is equivalent to $\rho=\rho^\ast$.

Similarly, 
\begin{equation*}
    E_{\rho^\ast}[\log(f_1(Y,S|\rho))|T_i=1]\leq E_{\rho^\ast}[\log(f_1(Y,S|\rho^\ast))||T_i=1],
\end{equation*}
where the equality holds if and only if $\rho=\rho^\ast$ or $\byba=0$.

Since at least one of $\byab$ and $\byba$ is not zero, then $q_0(\rho)$ is uniquely maximized on $\rho\in[\max(0,\psa+\psb-1), \min(\psa,\psb)]$ at $\rho^\ast$.

\item Uniform convergence:

Following in the same spirit as the proof for continuous intermediates in Appendix \ref{appendix:ProofOfThm1}, we only need to show that $\Var(\log(f_{T_i}(Y_i,S_i|\rho)|T_i=t)=\Var(\log f_t(Y(0),S(0)))$ and $E(\log(f_{T_i}(Y_i,S_i|\rho)|T_i=t)=E(\log f_t(Y(0),S(0)))$ is bounded w.r.t $\rho$ ($\pbb$) for $t=0,1$. It is sufficient for us to show that  $E\log^2f_t$ is bounded w.r.t $\rho$ ($\pbb$) for $t=0,1$.

For $t=0$, we have 
\begin{equation}
    \begin{split}
        \log(f_0)&=-\log(\sqrt{2\pi}\sigmay)-
        \log\Big(\exp\left(-\frac{(Y-\lambda_0-S\byaa)^2}{2\sigmay^2}\right) \paa^{1-S}\pba^{S}+\\
        &\qquad\exp\left(-\frac{(Y-\lambda_0-S\byaa-\byab)^2}{2\sigmay^2}\right)
        \pab^{1-S}\pbb^{S}\Big),
        \label{eq:binary:log marginal}
    \end{split}
\end{equation}
where conditional on $T=t$, $S=S(0)$ and $Y=Y(0)$.

The first term is constant, and we only need to show that the second moment of the second term is bounded w.r.t $\pbb$. For simplicity, we denote the expression inside the logarithm of the second term by $A+B$. Since $A+B<1$ (due to $\paa^{1-S}\pba^{S} + \pab^{1-S}\pbb^{S} \leq 1$ and $\exp(-(\cdot)^2)\leq 1$), $|log(A+B)|\leq |\log(2 \min(A,B))| \leq |\log(2 A)| + |\log(2 B)|$. Then, we only need to show the second moments of $|\log(A)$ and $|\log(B)|$ is bounded w.r.t $\pbb$.

The second moment of $|\log(A)|$ is 
\[
\begin{split}
    &E\left(-\frac{(Y(0)-\lambda_0-S(0)\byaa)^2}{2\sigmay^2}+(1-S(0))\paa+S(0)\pba\right)^2\leq \\
    &\qquad (1+1+1)\left(E\left(\frac{(Y-\lambda_0-S(0)\byaa)^2}{2\sigmay^2}\right)^2 + E(1-S(0))^2\paa^2+ ES^2(0)\pba^2\right).
\end{split}
\]

Thus, we only need to show that $E(Y^4(0))$ is bounded w.r.t $\pbb$. Similarly, for $|\log(A)|$, it is also sufficient to show that. Now, for $E(Y^4(0))$, we have 
\[
E(Y^4(0))=E\left(3 \sigmay^4 + \left(\lambda_0 + \byaa S(0) + \byab S(1)\right)^4\right),
\]
which is bounded w.r.t $\pbb\in[\max(0,\psa+\psb-1), \min(\psa,\psb)]$.

Similarly, we have the same results for $t=1$. Therefore, uniform convergence holds.
\end{itemize}
\end{proof}
Unlike for continuous intermediates, the marginal distribution is non-standard, and therefore it is infeasible to derive the closed form expression of the asymptotic approximation of the posterior variance.

\subsection{Identification up to Sign When Principal Ignorability Fails}
\label{appendix:BinaryPCEIdentificationIgnoringSigns}
For the continuous intermediates setting, we showed that even under a known association parameter, $(\byab,\byba)$ are unidentifiable and thus PCEs are also unidentifiable. These parameters are partially identified and we derived the partial identification region explicitly. In this section, we show that these same results do not hold for binary intermediates. For binary intermediates we will show that $(\byab,\byba)$ are actually point identified, but only up to sign. In other words, the magnitude is identified, though the signs of the coefficients are not identified and must be reasoned through using prior expertise. 

The marginal mean and variance (covariance) parameters of the observed outcome and intermediate can be written in terms of the joint model parameters, which was derived in (\ref{eq: constraints for joint prameters}), and is distribution-free and therefore also holds for binary intermediates. As a reminder, the general solutions of (\ref{eq: constraints for joint prameters}) are
    \begin{equation*}
    \begin{cases}
      &\beta_{t,1-t}^2=\dfrac{\zeta_t-\psi_t^2-\sigma_y^2}{\left(1-\Cor^2(S(0),S(1))\right)\sigma^2_{s,1-t}}\\
      &\beta_{tt}=\dfrac{\left(\psi_t-\Cor(S(0),S(1))\sigma_{s,1-t}\beta_{t,1-t}\right)}{\sigma_{st}}\\
      &\lambda_t=\mu_{yt}-(\beta_{t0}\phi_0+\beta_{t1}\phi_{1}),
    \end{cases} 
    \end{equation*}
where $(\mu_{yt},\phi_t)$ are the marginal means of $(Y(t), S(t))$ and 
\begin{equation*}
    \begin{pmatrix}
        \zeta_t&\psi_t\sigma_{st}\\
        \psi_t\sigma_{st}&\sigma_{st}^2
    \end{pmatrix}
\end{equation*}
is the marginal covariance matrix of $(Y(t), S(t))$ ($t=0,1$). These marginal parameters are identified as all mean parameters can be estimated using sample means, and variance parameters can be estimated using their sample-level analogs as well. 

Now, all remaining unknown parameters (except for the signs of ($\byab,\byba$)) can be expressed in terms of $\sigmay$ and the marginal parameters described above. We plug the solutions into conditional marginal pdf of $(Y,S)| T=0$ and obtain $f_0(y,s|\sigmay^2)$ (this can also be done with $f_1(y,s|\sigmay^2)=f_1(y,s|(\sigmay^\ast)^2)$). Next, we only need to show that $f_0(y,s|\sigmay^2)=f_0(y,s|(\sigmay^\ast)^2)$, for any $y\in\Rs$ and $s\in\{0,1\}$, implies that $\sigmay=\sigmay^\ast$. 

Suppose, for contradiction, that there is $\sigma_{y0}\in\Rs^+$ such that $f_0(y,s|\sigma_{y0}^2)=f_0(y,s|(\sigmay^\ast)^2)$, for any $y\in\Rs$ and $s\in\{0,1\}$. Note that 
\begin{equation}
    \begin{split}
    \dfrac{f_0(y,s|\sigmay^2)}{f_0(y,s|(\sigmay^\ast)^2)}&=\dfrac{\sigmay^\ast}{\sigmay} \exp\left(\dfrac{(y-\lambda_0^\ast-\byaa^\ast s)^2}{2{\sigmay^\ast}^2}-\dfrac{(y-\lambda_0-\byaa s)^2}{2\sigmay^2}\right)\\
    &\qquad \dfrac{1+\exp(-\dfrac{\byab^2-2\byab(y-\lambda_0-\byaa s)}{2\sigmay^2})\pab^{1-s}\pbb^{s}/(\paa^{1-s}\pba^{s})}{1+\exp(-\dfrac{(\byab^\ast)^2-2\byab^\ast(y-\lambda_0^\ast-\byaa^\ast s)}{2\sigmay^2})\pab^{1-s}\pbb^{s}/(\paa^{1-s}\pba^{s})}.
\end{split}
\label{eq:BinaryIdentificationEq}
\end{equation}

We first consider that $\byab<0$. Note that  
\[
\lim_{y\to\infty}\exp(-\dfrac{\byab^2-2\byab(y-\lambda_0-\byaa s)}{2\sigmay^2})\pab^{1-s}\pbb^{s}/(\paa^{1-s}\pba^{s})\to 0,
\]
for any $\sigmay\in\Rs^+$. For any $\delta>0$ and $\sigmay\in\{\sigma_{y0},\sigmay^\ast\}$, there is $y_0\in \Rs$, such that for $y>y_0$, we have
\[
\exp(-\dfrac{\byab^2-2\byab(y-\lambda_0-\byaa s)}{2\sigmay^2})\pab^{1-s}\pbb^{s}/(\paa^{1-s}\pba^{s})<\delta
\]
It follows that the second term of $f_0(y,s|\sigma_{y0}^2)/f_0(y,s|(\sigmay^\ast)^2)$, in (\ref{eq:BinaryIdentificationEq}), is bounded by $1/(1+\delta)$ and $1+\delta$ for $y>y_0$. Note that for the first term in (\ref{eq:BinaryIdentificationEq}), since $\sigma_{y0}\neq\sigmay^\ast$, then 
\[
\lim_{y\to\infty} \dfrac{\sigmay^\ast}{\sigmay} \exp\left(\dfrac{(y-\lambda_0^\ast-\byaa^\ast s)^2}{2(\sigmay^\ast)^2}-\dfrac{(y-\lambda_0-\byaa s)^2}{2\sigma_{y0}^2}\right)=0.
\]
For any $\delta>0$, there is $y_1\in\Rs$, such that for $y>y_1$, 
\[
\dfrac{\sigmay^\ast}{\sigmay} \exp\left(\dfrac{(y-\lambda_0^\ast-\byaa^\ast s)^2}{2(\sigmay^\ast)^2}-\dfrac{(y-\lambda_0-\byaa s)^2}{2\sigma_{y0}^2}\right)<\delta.
\]

Therefore, considering $\delta=1/2$, for $y>\max(y_0,y_1)$,
\[
0<f_0(y,s|\sigma_{y0}^2)/f_0(y,s|(\sigmay^\ast)^2)<\delta (1+\delta)=3/4<1,
\]
which contradicts the fact that $f_0(y,s|\sigma_{y0}^2)=f_0(y,s|(\sigmay^\ast)^2)$, for any $y\in\Rs$ and $s\in\{0,1\}$. This indicates that $\sigmay=\sigmay^\ast$ and $\sigmay$ is identifiable. Moreover, $\byab$ and $\byba$ are also identifiable (except for the signs). The reason that the signs are not identified or are weakly identified can be seen from the first equation in (\ref{eq: constraints for joint prameters}), where knowing $\sigma_y$ only identifies $\beta_{t,1-t}^2$.

\subsection{Simulations with binary intermediate}
Consider the following data generating process:
\begin{equation*}
    \begin{split}
        Y_i(0)\mid \bd{U}_i&\sim N(0.9+1.2S(0)+0.6S(1),0.5^2),\\
        Y_i(1)\mid \bd{U}_i&\sim N(0.5+0.8S(0)+1.2S(1),0.5^2),\\
        P(\bd{U}_i=(\uia,\uib))&= 0.1^{(1-\uia)(1-\uib)}0.3^{(1-\uia)\uib}0.2^{\uia(1-\uib)}0.4^{\uia\uib}.
    \end{split}
\end{equation*}

We generate one large dataset with $n = 5,000$ to investigate the above identification results for $\byab$ and $\byba$. For this dataset, we treat $\pbb$ as known and run MCMC both with and without a constraint that these coefficients are necessarily positive. Gibbs sampling is used, with a single MCMC chain run for a total of 8,000 iterations, where the first 2,000 are burn-in, and a thinning interval of 20 is applied. The details of MCMC sampling are included in the following subsection. We consider noninformative conjugate priors as follows:
\begin{equation*}
    \begin{split}
        &\bbeta_t\sim N(0,10^5\bd{I}_2),\ 
        \lambda_t\sim N(0,10^5),\ \sigma_y^2\sim IG(10^{-3},10^{-3}),\ \psa\sim U(0,1),\ \psb\sim U(0,1),
    \end{split}
\end{equation*}
where $IG(a,b)$ represents the inverse gamma distribution and $\bd{I}_2$ is a $2\times2$ identity matrix. 

The trace plots of $\byab$ and $\byba$, both with and without the positive sign constraint, are shown in Figures \ref{fig: traceplots of beta01 sign} and \ref{fig: traceplots of beta10 sign}, respectively. Without the sign constraint, the trace plots of $\byab$ and $\byba$ tend to explore different signs, and the trace plot of $\byba$ bounces around the negative of the true value. This highlights our result that these parameters are only identified up to sign, and therefore without a sign constraint, the MCMC algorithm bounces between different values that are equally supported by the data. After applying the positive (true) sign constraint, both trace plots hover around the true values, confirming the identification results for $\byab$ and $\byba$ discussed in Appendix \ref{appendix:BinaryPCEIdentificationIgnoringSigns}. 

\begin{figure}[tb]
    \centering
    \includegraphics[width=0.75\linewidth]{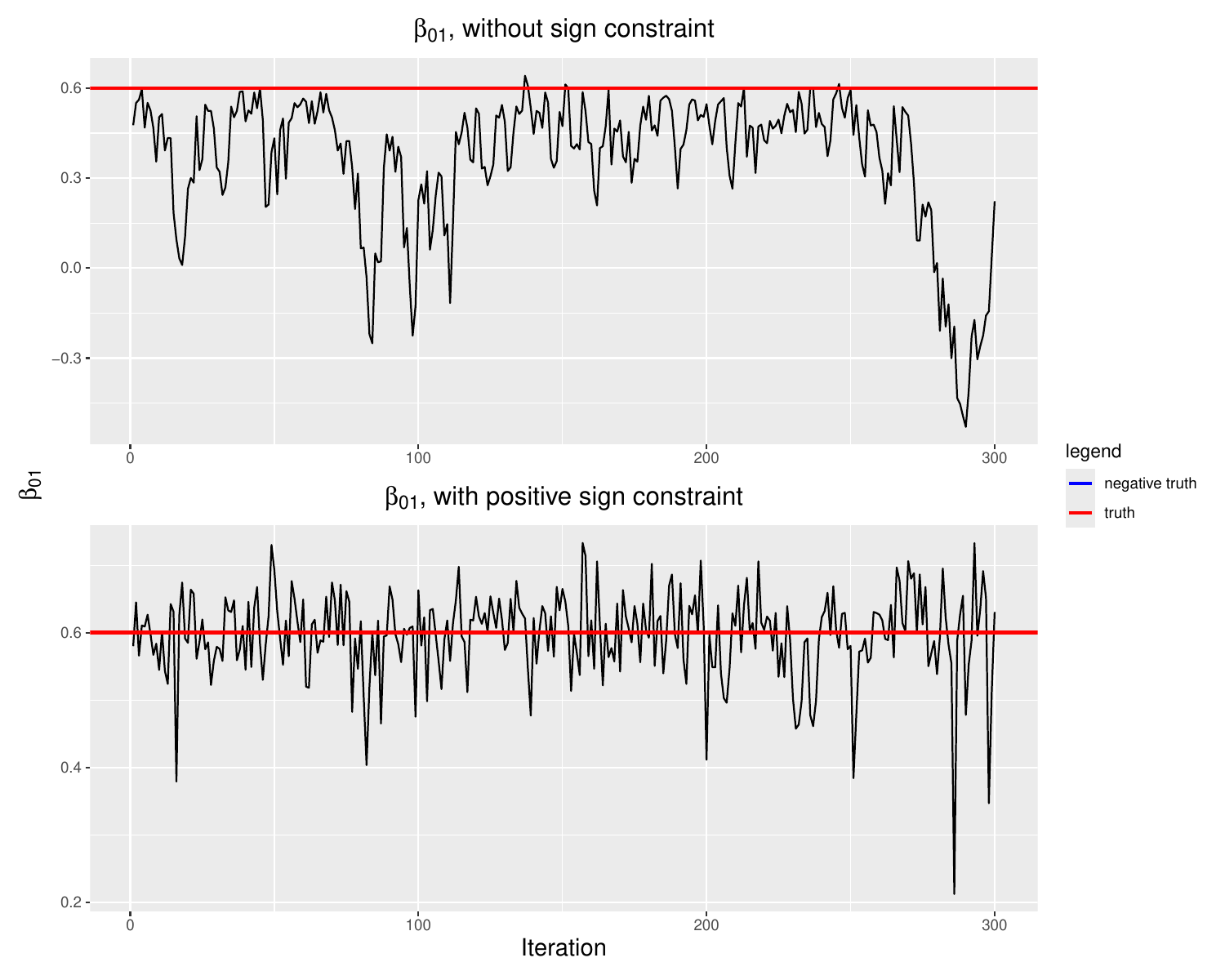}
    \caption{Trace plots of $\byab$ with and without the positive sign constraint.}
    \label{fig: traceplots of beta01 sign}
\end{figure}

\begin{figure}[tb]
    \centering
    \includegraphics[width=0.75\linewidth]{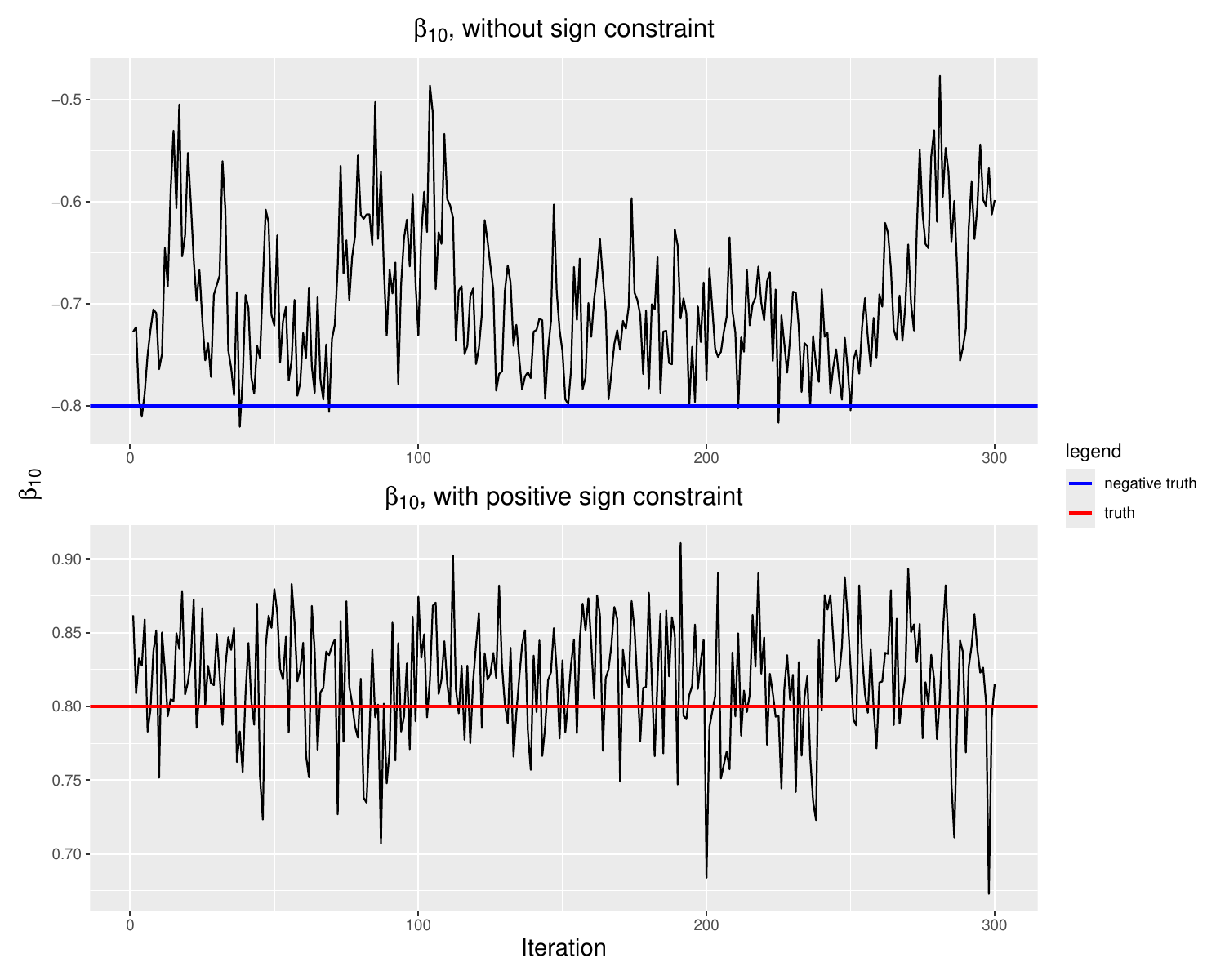}
    \caption{Trace plots of $\byba$ with and without the positive sign constraint.}
    \label{fig: traceplots of beta10 sign}
\end{figure}

Additionally, We also generate one large dataset with $n = 25,000$ to investigate whether $\pbb$ is also identified. We again fit the models both with
and without the positive sign constraint, and use the same MCMC configuration as in the previous simulated dataset. However, we now treat $\pbb$ as unknown and place a flat prior distribution for this parameter. The trace plots of $\pbb$ with and without the positive sign constraint are shown in Figure \ref{fig: traceplots of p11 identification}. Both trace plots of $\pbb$ hover around the true value, suggesting that $\pbb$ is identified as well when the outcome model parameters are identified. It is worth noting, however, that all results in this subsection hold in the absence of covariates and further research is needed to determine whether they hold more generally.

\begin{figure}[tb]
    \centering
    \includegraphics[width=0.75\linewidth]{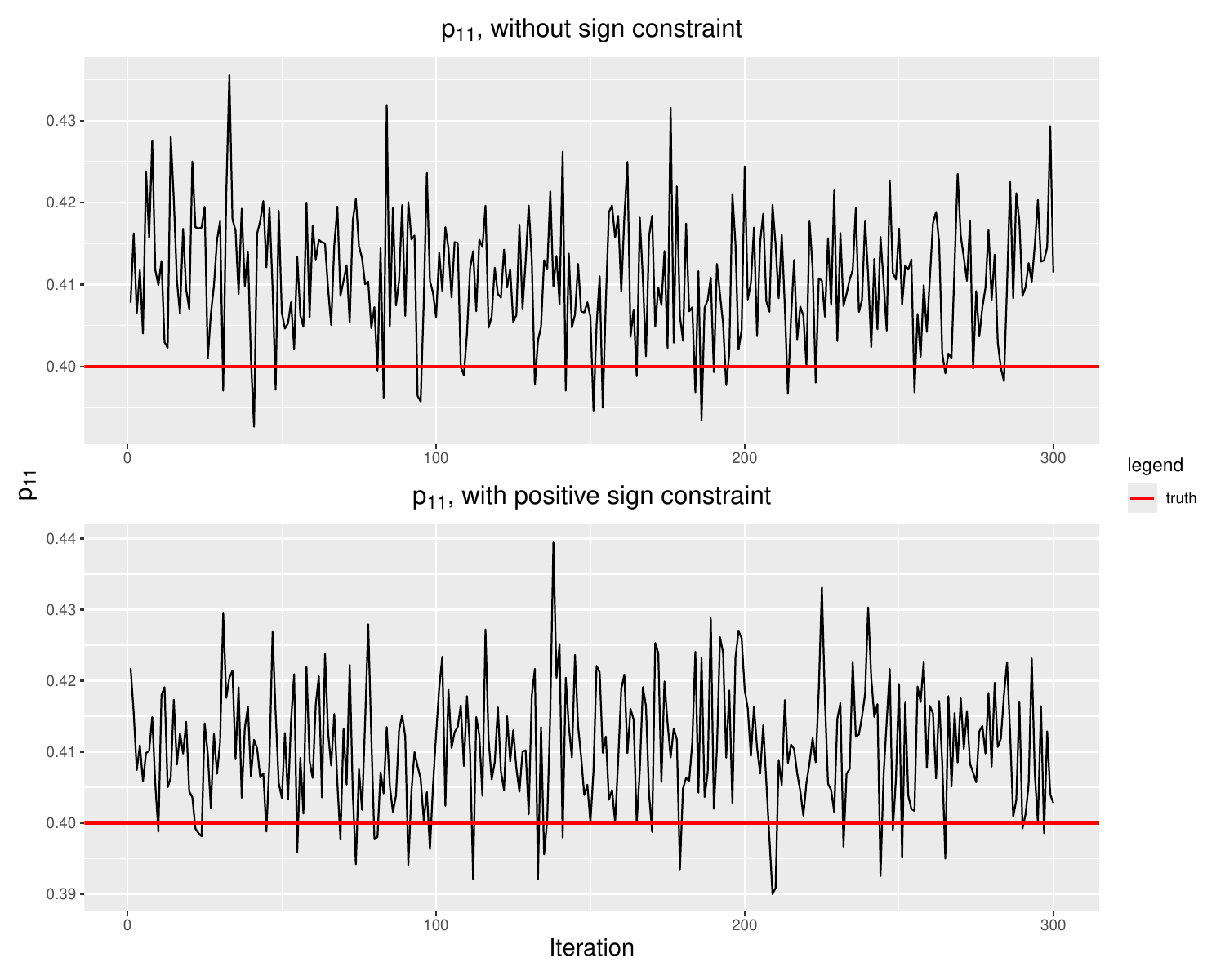}
    \caption{Trace plots of $\pbb$ with and without the positive sign constraint.}
    \label{fig: traceplots of p11 identification}
\end{figure}

\subsection{Details of MCMC Sampling for binary intermediates}
In this subsection, let $\btheta_y=(\bbeta_0,\bbeta_1,\lambda_0,\lambda_1)$, which is column vector. The full joint pdf of is: 
\[
\begin{split}
    \dfrac{1}{(2\pi)^{\frac{n}{2}}\sigmay^n}\exp\left(-\suml{i=1}^n\frac{(y_i-\btheta_y^T\bd{d}_i)^2}{2\sigmay^{2}}\right)\paa^{n_{00}}\pab^{n_{01}}\pba^{n_{10}}\pbb^{n_{11}},
\end{split}
\]
where $n_{00}=\suml{i}(1-\uia)(1-\uib)$, $n_{01}=\suml{i}(1-\uia)\uib$, $n_{10}=\suml{i}\uia(1-\uib)$, $n_{11}=\suml{i}\uia\uib$, and 
\[
\bd{d}_i=\begin{pmatrix}
            \bd{u}_i(1-t_i)\\
            \bd{u}_it_i\\
            1-t_i\\
            t_i
        \end{pmatrix}.
\]

We consider conjugate priors for all parameters except $(\btheta_s,\pbb)$ as follows: 
\[
\bbeta_t\sim N(\bmu_{\bbeta_t},\Sigma_{\bbeta_t}),\ 
        \lambda_t\sim N(\mu_{\lambda_t},\sigma_{\lambda_t}^2),\ \sigma_y^2\sim IG(\eta_y,\nu_y),
\]
and flat priors for $(\psa,\psb,\pbb)$.

Then, the updating steps for Gibbs sampling are as follows:
\begin{enumerate}
    \item We can update the unobserved values of the intermediate $S(1-t)$ from a Bernoulli with probability 
    \begin{equation*}
         \dfrac{h_{t}(s_{t},1)}{h_{t}(s_{t},1)+h_{t}(s_{t},0)},
    \end{equation*}
    where 
    \begin{equation*}
         h_t (s_t,s_{1-t}) = \exp(-\frac{(y-\lambda_{t}-\bbeta_{t}^T\bd{u})^2}{2\sigmay^{2}})\paa^{(1-s_0)(1-s_1)}\pab^{(1-s_0)s_1}\pba^{s_0(1-s_1)}\pbb^{s_0s_1}.
    \end{equation*}

    \item We can update $\bd{\theta}_y$ from the following conditional distribution
    {\small\begin{equation*}
    \bd{\theta}_y\mid \cdot\sim N\left(\bigg(\suml{i=1}^n\bd{d}_{i}\bd{d}_i^T/\sigma_y^2+\Sigma_{\btheta_y}^{-1}\bigg)^{-1}\bigg(\suml{i=1}^ny_i\bd{d}_i/\sigma_y^2+\Sigma_{\btheta_y}^{-1}\bmu_{\btheta_y}\bigg),\bigg(\suml{i=1}^n\bd{d}_{i}\bd{d}_i^T/\sigma_y^2+\Sigma_{\btheta_y}^{-1}\bigg)^{-1}\right),
    \end{equation*}}
    where 
    \begin{equation*}
        \bSigma_{\btheta_y}=\text{diag}(\bSigma_{\bbeta_0},\bSigma_{\bbeta_1},\sigma^2_{\lambda_0},\sigma^2_{\lambda_1}).
    \end{equation*}

    \item We can update $\sigmay$ from
    \begin{equation*}
        \sigma_y^2\mid \cdot\sim IG\left(\frac{n}{2}+\etay,\suml{i=1}^n(\btheta_y^T\bd{d}_i-y_i)^2)/2+\nuy\right).
    \end{equation*}
    \item 
    Since 
    \[
    \begin{split}
        p(\psa,\psb)&\propto \prod_i\Big((1+\pbb-\psa-\psb)^{(1-\uia)(1-\uib)}(\psb-\pbb)^{(1-\uia)\uib}\\
        &\qquad (\psa-\pbb)^{\uia(1-\uib)}\pbb^{\uia \uib}\Big),
    \end{split}
    \]

    then we first update $(\psa-\pbb,\psb-\pbb)$ from
    \[
    (\psa-\pbb,\psb-\pbb,1+\pbb-\psa-\psb)/(1-\pbb)\sim \text{Dirichlet($n_{10}+1,n_{01}+1,n_{00}+1$)},
    \]
    and next, obtain $(\psa,\psb)$.
    
    \item 
    Instead of updating $\pbb$ from posterior of full data including unobserved intermediates, we update $\pbb$ using marginal posterior with observed intermediate variables only. Then, we have
    \[
    p(\pbb\mid \cdot)\propto \prod_i \left(f_1^{t_i}(y_i,s_{i}|\pbb)f_0^{1-t_i}(y_{i},s_{i}|\pbb)\right).
    \]
    
    Since the conditional distribution of $\pbb$ is intractable, we use a grid-based posterior sampling to update $\pbb$: we evaluate the posterior over a fine grid of values and sample from the normalized grid-based posterior. 
\end{enumerate}

\section{Additional Plots}
\label{subsec:AdditionalPlots}
We ran additional MCMC chains under the setting described in Section \ref{subsec:Scenario2} with $n$ varying from 300 to 9600.
We plotted $\log(n)$ versus the logarithm of the estimated posterior variance, as shown in Fig. \ref{fig:ConvergenceRatePlot}. This plot shows that when the sample size is large (over 1000),  it is nearly linear, indicating that the posterior variance of $\rho$ is approximately $O(n^{-1})$, consistent with the rate stated in Theorem \ref{thm:ApproximateVariance}.

\begin{figure}[tp]
    \centering
    \includegraphics[width=0.5\linewidth]{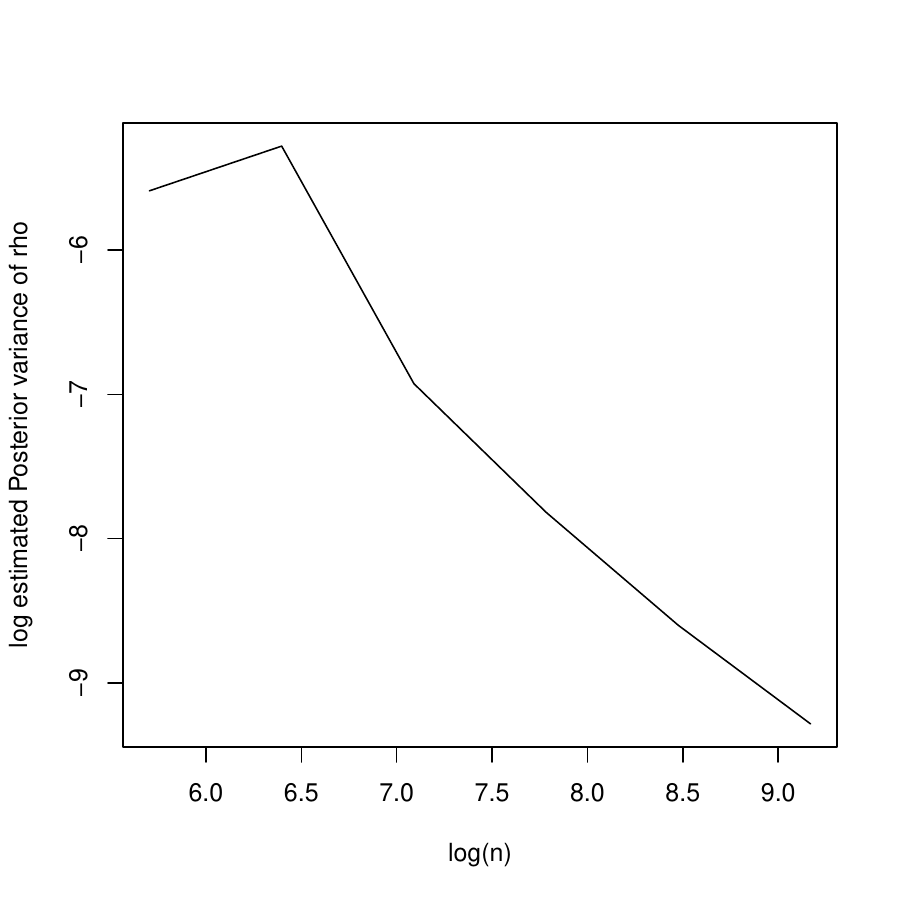}
    \caption{ $\log(n)$ vs. logarithm of the estimated posterior variance}
    \label{fig:ConvergenceRatePlot}
\end{figure}

\end{appendix}

\bibliographystyle{imsart-number} 
\bibliography{my}       


\end{document}